\documentclass[a4paper,11pt]{article}
\pdfoutput=1

\usepackage{jheppub}
\usepackage{tikz}
\usetikzlibrary{shapes}
\usetikzlibrary{plotmarks}

\usepackage{mathrsfs}
\usetikzlibrary{arrows}

\definecolor{light}{rgb}{0,1,1}

\DeclareMathOperator{\coker}{coker}
\DeclareMathOperator{\sdet}{sdet}

\DeclareMathOperator{\vt}{Vect}
\DeclareMathOperator{\chara}{character}
\DeclareMathOperator{\IM}{Im}

\newcommand{\Q}{\mathcal{Q}}
\renewcommand{\L}{\mathcal{L}}
\newcommand{\R}{\mathcal{R}}
\newcommand{\e}[1]{\mathrm{e}^{#1}}

\title{Index of the Transversally Elliptic Complex in Pestunization}
\author{Roman Mauch}
\author{and Lorenzo Ruggeri}
\affiliation{Department of Physics and Astronomy, Uppsala Universitet\\Lägerhyddsvägen 1, 752 37 Uppsala, Sweden}
\emailAdd{roman.mauch@physics.uu.se}
\emailAdd{lorenzo.ruggeri@physics.uu.se}
\preprint{UUITP-64/21}

\abstract{ In this note we present a formula for the equivariant index of the cohomological complex obtained from localization of $\mathcal{N}=2$ SYM on simply-connected compact four-manifolds with a $T^2$-action. Knowledge of said index is essential to compute the perturbative part of the partition function for the theory. In the topologically twisted case, the complex is elliptic and its index can be computed in a standard way using the Atiyah-Bott localization formula. Recently, a framework for more general types of twisting, so-called cohomological twisting, was introduced for which the complex turns out to be only transversally elliptic. While the index of such a complex has been computed for some cases where the manifold can be lifted to a Sasakian $S^1$-fibration in five dimensions, a general four-dimensional treatment was still lacking. We provide a formal, purely four-dimensional treatment of the cohomological complex, showing that the Laplacian part can be globally split off while the remaining part can be trivialized uniquely in the group-direction. This ultimately produces a simple formula for the index applicable for any compact simply-connected four-manifold. Finally, the index formula is applied to examples on $S^4$, $\mathbb{CP}^2$ and $\mathbb{F}^1$. For the latter, we use the result to compute the perturbative partition function.}

\begin{document}

\maketitle
\flushbottom

\section{Introduction}
    The understanding of supersymmetric quantum field theories on compact manifolds has benefited widely from localization techniques, starting with the works in \cite{Witten:1988ze,Nekrasov:2002qd,arXiv:0712.2824}. Following these results, great progress has been made towards extending the localization procedure to different dimensions, background geometries and number of supercharges. A comprehensive review is provided in \cite{Pestun:2016zxk}. 
    
    In this note, we consider $\mathcal{N}=2$ SYM theories on compact, simply-connected four-manifolds $X$ which are equipped with a $T^2$-action, generated by a Killing vector field with isolated fixed points. Many results have been obtained for this setup in the literature \cite{arXiv:0712.2824,Bawane:2014uka,Bershtein:2015xfa,Bershtein:2016mxz,Hama:2012bg,Lundin:2021zeb}. Relying on these, it has been conjectured in \cite{arXiv:1812.06473} that an arbitrary distribution of Nekrasov partition functions for either instantons or anti-instantons (corresponding to anti-self-dual (ASD) or self-dual (SD) connections) at each fixed point gives rise to the partition function of a valid supersymmetric theory. This was extended to theories involving matter in \cite{Festuccia:2020yff}. Taking $S^4$ as an example, by distributing ASD at both poles we obtain equivariant Donaldson-Witten theory \cite{Witten:1988ze} counting ASD connections on $S^4$. Placing ASD at one pole and SD at the other pole instead gives Pestun's theory on $S^4$ \cite{arXiv:0712.2824}. Having SD at both poles gives again equivariant Donaldson-Witten theory, now counting SD connections. 
    
    It has been shown in \cite{arXiv:1904.12782} that for cohomological twisting (as opposed to topological twisting which is a special case of the former) on $X$, the localization procedure naturally gives rise to a double-(cochain)complex with the maps provided by supersymmetry and BRST transformations. Moreover, the one-loop determinant resulting from localization can be computed from the equivariant index of the horizontal component of the double-complex. The latter is of a standard form (given in \eqref{eq--complex}) and is elliptic for topological twisting, but this is no longer true in the more general case; rather, the complex is elliptic only transverse to the $T^2$-action. While the equivariant index for the elliptic case can be computed straightforwardly using the Atiyah-Bott formula \cite{Atiyah:1967}, the computation for the transversally elliptic complex turns out to be more subtle. This is due to the fact that the cohomologies of such a complex are not finite-dimensional anymore (which would be the case for an elliptic complex, which is Fredholm for $X$ compact). However, the cohomology at each level can be decomposed into irreducible representations of the $T^2$-action, each of them appearing with finite multiplicity \cite{Atiyah:1974}. Therefore, the index becomes a distribution on $T^2$ rather than an ordinary function (as in the elliptic case).
    
    The index computation for such cases has been performed, for example, in \cite{arXiv:0712.2824,Hama:2012bg} on $S^4$ and more generally in \cite{arXiv:1812.06473,arXiv:1904.12782} for manifolds that can be lifted to a Sasakian $S^1$-fibration in five dimensions with a specific type of ASD/SD distributions (which correspond to different $S^1$-fibrations over $X$). It turns out that the index is still composed of the elliptic contributions around the fixed points, however, they have to be regularized in an suitable way\footnote{One might like to think of the different ways to regularize as the different ways to ``glue'' the fixed point contributions together.}.  
    
    This work provides an extension to the aforementioned index computation to any simply-connected, compact four-manifold and arbitrary distributions of ASD/SD at the fixed points, for the zero-flux sector. This is done essentially by decomposing the symbol of the transversally elliptic complex, denoted $(E^\bullet,D)$ in the following and given in \eqref{eq--complex}, into more accessible parts. We achieve this by using the fact that the index depends on the symbol only up to homotopy and, thus, only the corresponding symbol class is relevant. We first show that, up to an isomorphism, the symbol can be split into two parts, an elliptic one (which is simply the Laplacian) and a transversally elliptic one. Using the homomorphism property of the index, both parts can be computed separately and the Laplacian part does not contribute. For the remaining, transversally elliptic part we make use of a filtration of $X$ with respect to the $T^2$-action, presented in \cite{Atiyah:1974}. This allows to construct a new symbol (homotopic to the original one) by essentially taking the fixed point contributions of the original symbol and gluing them along the $T^2$-action in a compatible way (this is discussed in \autoref{sec--2}) such that the new symbol is only supported at the fixed points. Finally, the index decomposes into contributions from the fixed points whose regularization turns out to be determined uniquely\footnote{Up to some ambiguity arising from gluing which does not affect the index itself.} by the ASD/SD distribution at the fixed points. 
    
    Once the equivariant index of $(E^\bullet,D)$ is determined, it can be used to compute the one-loop determinant arising from localization, which itself constitutes the (exact) perturbative part of the partition function for the $\mathcal{N}=2$ SYM theory under consideration. We expect our procedure to extend also to the case of non-zero flux  and to theories including matter and shall comment on this later on.
    
    The article is organized as follows. In \autoref{sec--1} we briefly recall how localization gives rise to a complex and how one can determine the one-loop determinant in terms of the index of this complex. Subsequently, we show that the symbol of the complex splits globally into a Laplacian part (which is elliptic) and a transversally elliptic part (with respect to the $T^2$-action). In \autoref{sec--2} we ``break up'' the symbol into simpler pieces, namely the (elliptic) contributions at the fixed points and give a prescription on how to glue these together in a way that the new symbol is homotopic to (hence gives the same index than) the original one. Finally, in \autoref{sec--3} we compute the index of the symbol and provide an explicit formula \eqref{eq--eqindex} which can be used to determine the one-loop determinant of the theory at hand. We exemplify our procedure in \autoref{sec--4} for different ASD/SD distributions on $S^4,\mathbb{CP}^2$ and $\mathbb{F}^1$. For the cases known from \cite{arXiv:1812.06473} we check that they are in agreement with our results, but also provide some new examples which cannot be computed from \cite{arXiv:1812.06473}. We conclude by applying the index computation for $\mathbb{F}^1$ to find the perturbative part of the partition function in the zero-flux sector.

    \subsection{Summary of Procedure}
    
        For the physically-oriented reader, we provide a quick recipe on how to compute the equivariant index of \eqref{eq--complex} for a given manifold $X$ and distribution of SD/ASD complexes\footnote{Note that for an ASD connection $A$ we have $F^+_A=0$ and thus the SD complex is the relevant one. Hence, when we talk about a SD complex over some fixed point this is associated with an instanton (ASD connection) while an ASD complex is associated with an anti-instanton (SD connection). In the remainder, when we say SD/ASD we always refer to the complex.} over the set of torus fixed points $Y$. The full index is obtained as a product of the character of the adjoint representation of the gauge group, $\chi_\mathrm{Ad}$ with the index of the isometry-part \eqref{eq--complexH} of the complex, denoted $\iind\eth$. The latter can be determined through the following steps:
        \begin{itemize}
            \item[(i)] Introduce an atlas on $X$ such that each fixed point $l$ is contained in one patch $U_l$, with local complex coordinates $(z^{(l)}_1,z^{(l)}_2)$.
            \item[(ii)] On $U_l$, denote by $\epsilon_1,\epsilon_2$ the coordinates on $\Lie T^2$ and $t_{1}=\exp(\mathrm{i}\epsilon_{1}),\;t_2=\exp(\mathrm{i}\epsilon_{2})$ the corresponding coordinates on $T^2$. For a $T^2$-action on $(z^{(l)}_1,z^{(l)}_2)$ given by
            \begin{equation*}
                 z^{(l)}_1\mapsto t_1^{\alpha^{(l)}_{11}}t_2^{\alpha^{(l)}_{12}}z^{(l)}_1,\qquad z^{(l)}_2\mapsto t_1^{\alpha^{(l)}_{21}}t_2^{\alpha^{(l)}_{22}}z^{(l)}_2,\qquad\alpha^{(l)}_{ij}\in\mathbb{Z},
            \end{equation*}
            read out the infinitesimal weights:
            \begin{equation*}
                \alpha^{(l)}_i=(\alpha^{(l)}_{i1},\alpha^{(l)}_{i2})
                ,\quad i=1,2.
            \end{equation*}
            If an ASD complex is placed at $l$, perform the flip $\alpha^{(l)}_1\mapsto -\alpha^{(l)}_1$. 
            \item[(iii)] Create an array $s=(s_1,\dots,s_n)$, with $n$ the number of fixed points, in which $s_l=+$ for SD at $l$ or $s_l=-$ for ASD.
            \item[(iv)] Finally, insert everything into the following formula for the index: 
            \begin{equation*}\label{eq--eqindex.intro}
                \iind\eth=-\sum_{l\in Y}\left(1+\prod_{i=1}^2t^{-\alpha^{(l)}_i}\right)\prod_{k=1}^2\left(\frac{1}{1-t^{-\alpha^{(l)}_k}}\right)^{s_l},
            \end{equation*}
            with $t^{\alpha^{(l)}_i}:=\prod_{j=1}^2t_j^{\alpha^{(l)}_{ij}}$ and $(\cdot)^{\pm}$ the Laurent expansion at $t=0$, $t=\infty$, respectively (cf. \eqref{eq--laurent.t-1}).
        \end{itemize}
        The full index, $\iind\eth\cdot\chi_\mathrm{Ad}$ (with $\chi_\mathrm{Ad}$ the character of the adjoint representation of $G$) can then be used to determine the perturbative partition function as outlined in \autoref{sec--1} and exemplified in \autoref{sec--4} for $\mathbb{F}^1$. 

\section{Transversally Elliptic Complex from Localization}\label{sec--1}

    This work is concerned with cohomologically twisted $\mathcal{N}=2$ SYM theories on simply-connected compact Riemannian four-manifolds, always denoted by $X$ in the following. It is further assumed that the theory is invariant under a $T^2$-isometry of $X$ generated by a Killing vector field with isolated fixed points. The general theory of this setup has been introduced in \cite{arXiv:1812.06473,arXiv:1904.12782} which provide the basis of this work and some familiarity with these references is essential for a good understanding of it.
    
    In this section, we first recall how an index computation naturally arises in the process of localization in such theories and present the complex of which the index needs to be computed. Subsequently, we prove that this complex is transversally elliptic and give a precise formulation of the index computation.

    \subsection{Cohomological Complex from $\mathcal{N}=2$ Localization}
    
        In the cohomologically twisted setting, after appropriate gauge fixing was performed, the field content schematically is given by even fields $\Phi,\hat{\Phi}$ and odd ones, $\Psi,\hat{\Psi}$. They are related, via a supercharge $\Q$ (which really is a combination of supersymmetry and BRST now) by the following transformations:
        \begin{equation}\label{eq--SUSY}
            \begin{aligned}
                \Q\Phi=\Psi,&\qquad \Q\Psi=\R\Phi,\\
                \Q\hat{\Psi}=\hat{\Phi},&\qquad \Q\hat{\Phi}=\R\hat{\Psi}.
            \end{aligned}
        \end{equation}
        Here, $\R$ is the action of the symmetry groups, including the gauge group $G$ and the isometry group of $X$ containing $H=T^2$ as a subgroup. We denote the symmetry groups collectively by $R$. From \eqref{eq--SUSY} we see that $(\Phi,\Psi)$ and $(\hat{\Psi},\hat{\Phi})$ form multiplets with $\Q^2=\R$, i.e. $\Q$ can be viewed as an equivariant differential. In order to perform localization, the action is deformed by a $\Q$-exact term $\delta S=t\int\Q V$ ($t\in\mathbb{R}$) such that $\Q^2 V=0$ (i.e. $V$ is $\R$-invariant) and the bosonic part of $\Q V$ is positive semi-definite. By the standard argument, the path integral of the deformed action is independent of the value of $t$ and we choose $t\to\infty$. The path integral then localizes to field configurations such that $\Q V=0$. 
        A common choice for $V$ is $\langle \mathcal{F},\Q \mathcal{F}\rangle$, where $\mathcal{F}$ denotes the fermionic field content and $\langle\cdot,\cdot\rangle$ an inner product invariant under the symmetries $R$ with reality conditions for the fields chosen suitably such that $\langle\Q\mathcal{F},\Q\mathcal{F}\rangle$ is positive semi-definite. The one-loop part of the partition function is obtained by expanding $\Q V$ around these configurations to second order, producing a quadratic term
        \begin{equation}\label{eq--V2}
            \Q V^{(2)}=\Q\left(\langle\Psi,\R\Phi\rangle+\langle\hat{\Psi},D\Phi\rangle+\langle\hat{\Psi},\hat{\Phi}\rangle\right),
        \end{equation}
        where $\Phi,\Psi,\hat{\Phi},\hat{\Psi}$ now denote fluctuations around the localization locus. This quadratic term can be integrated to the one-loop determinant (cf. \cite{arXiv:1904.12782} section 2)
        \begin{equation*}
            \frac{\det^{1/2}|_{\coker D}\R}{\det^{1/2}|_{\ker D}\R}=\sdet^{1/2}|_{H^\bullet(D)}\R,
        \end{equation*}
        where the numerator arises from fermionic integration and the denominator from bosonic one. The operator $D$ is the piece in the quadratic action that ``maps $\Phi$ to $\hat{\Psi}$''\footnote{It was shown in \cite{arXiv:1904.12782} that supersymmetry and BRST differentials form a double complex, of which $D$ denotes the horizontal maps.} and $H^\bullet(D)$ is its cohomology. We can now decompose $\ker D$ and $\coker D$ into irreducible representations of the group action of $\R$, labelled by $\alpha$, which appear with finite multiplicity $m_\alpha^\mathrm{ker}$ and $m_\alpha^\mathrm{coker}$. Note that for elliptic $D$, the number of irreducible representations is always finite (remember $X$ is compact, i.e. $D$ is Fredholm) whereas for transversally elliptic $D$, it is, in general,  infinite \cite{Atiyah:1974}. Let $w_\alpha(\epsilon)$ denote the (sum of) weights of representation $\alpha$ depending on the equivariant parameters, collectively denoted by $\epsilon$ (which can be read off from $\mathcal{L}$). Then we can write
        \begin{equation}\label{eq--sdet}
            \sdet^{1/2}|_{H^\bullet(D)}\R=\prod_\alpha w_\alpha(\epsilon)^{(m_\alpha^\mathrm{coker}-m_\alpha^\mathrm{ker})/2}.
        \end{equation} 
        The weights and multiplicities can be extracted from the equivariant index of $D$,
        \begin{equation}\label{eq--iindD}
            \iind D=\sum_\alpha (m_\alpha^\mathrm{ker}-m_\alpha^\mathrm{coker})\mathrm{e}^{w_\alpha(\epsilon)},
        \end{equation}
        where $\mathrm{e}^{w_\alpha(\epsilon)}$ denotes the respective character.
        Therefore, in order to compute the one-loop determinant it is paramount to have the equivariant index of $D$ at our disposal.

        For cohomological twisting, the field strength $F$ of the gauge field localizes to anti-self-dual configurations, $F\in P^+\Omega^2_X$, at some fixed points of the $H$-action and to self-dual ones, $F\in P^-\Omega^2_X$, at the remaining fixed points. Hence, a global description requires the introduction of a generalized projector $P^+_\omega$ that interpolates between $P^+$ and $P^-$ away from the fixed points. The construction of $P^+_\omega$ basically consists of gluing spaces $\Omega^{2+}$ and $\Omega^{2-}$ on the overlap of the respective patches around the fixed points using the isomorphism
        \begin{equation}\label{eq--transition}
            m:\Omega^{2+}\rightarrow\Omega^{2-},\,\beta\mapsto -\beta+\frac{2}{\|v\|^2}\kappa\wedge\iota_v\beta.
        \end{equation}  
        Here, $v$ denotes the Killing vector field that generates the $H$-action and $\kappa=g(v,\cdot)$ is the one-form canonically associated to $v$ via the metric $g$. A detailed derivation of $P^+_\omega$ is provided in \cite{arXiv:1812.06473} and here we only state the result,
        \begin{equation}\label{eq--projector}
            P^+_\omega=\frac{1}{1+\cos^2\omega}\left(1+\cos\omega\star-\sin^2\omega\,\frac{\kappa\wedge\iota_v}{\|v\|^2}\right),
        \end{equation}
        with $\cos\omega=1$ at some fixed points and $\cos\omega=-1$ at the remaining ones (e.g. $\cos\omega=1$ at all fixed points for Donaldson-Witten theory).

        Finally, we need to know the explicit form of $D$ in order to compute the equivariant index. We have $\Phi=(A,\varphi)$ and $\Psi=(\chi,\bar{c},c)$, with gauge field $A$, real scalar $\varphi$, a two-form $\chi$ in the image of $P^+_\omega$ and ghost fields $c,\bar{c}$. 
        From the localization action in \cite{arXiv:1812.06473} we obtain
        \begin{equation}\label{eq--Donbosons}
            V^{(2)}\supset c\wedge\star\mathrm{d}^\dagger(\iota_vF+\mathrm{d}(\cos\omega\,\varphi))+\bar{c}\wedge\star\mathrm{d}^\dagger A+\chi\wedge\star P^+_\omega(F+\iota_v\star\mathrm{d}\varphi),
        \end{equation}
        where we chose Lorenz gauge for the gauge-fixing. 
        Here, $A\in\Omega^1_X(\mathfrak{g}_P),\varphi\in\Omega^0_X(\mathfrak{g}_P)$ are gauge and scalar field fluctuations around the localization locus. $\mathfrak{g}_P=P\times_\mathrm{Ad}\Lie G$ denotes the associated bundle of Lie algebras to the principal bundle $P$ over $X$ via the adjoint action of the gauge group $G$ (i.e. $A,\varphi$ are Lie algebra-valued one-forms and scalar fields as usual). Note that \eqref{eq--Donbosons} is obtained by expanding around the trivial connection for an Abelian gauge group. It is easily extended to the general non-Abelian case, but the index only depends on the (principal) symbol of $D$ to which the additional terms that would appear do not contribute. Note that in this setup, the square of the supercharge is $\R=\mathrm{i}\L_v+G_{a_0}$ with $v$ the Killing vector and $G_{a_0}$ a gauge transformation by the Coulomb parameter $a_0$. 

        The actions \eqref{eq--Donbosons} on the bosonic fields can be recast into a cochain complex
        \begin{equation}\label{eq--complex}
            \begin{tikzcd}[column sep=scriptsize]
                0\ar[r] & (\Omega^1_X\oplus\Omega^0_X)\otimes\Gamma(\mathfrak{g}_P)\ar[r,"\eth\otimes 1"] & (P^+_\omega\Omega^2_X\oplus\Omega^0_X\oplus\Omega^0_X)\otimes\Gamma(\mathfrak{g}_P)\ar[r] & 0
            \end{tikzcd},
        \end{equation}
        where $\Gamma(\mathfrak{g}_P)$ denotes the sheaf of sections on $\mathfrak{g}_P$ and the differential operator $\eth$ is given by
        \begin{equation}\label{eq--pdo}
            \eth=\begin{pmatrix}
                P^+_\omega\mathrm{d} & \phantom{-}P^+_\omega\iota_v\star\mathrm{d}\\
                \mathrm{d}^\dagger\iota_v\mathrm{d} & \phantom{-}\mathrm{d}^\dagger\mathrm{d}\cos\omega\\
                \mathrm{d}^\dagger & 0
            \end{pmatrix}.
        \end{equation}
        We denote this complex by $(E^\bullet,D)$ in the following (with $D=\eth\otimes 1$). Note that \eqref{eq--complex} is $H$-invariant\footnote{For $h\in H$, denote its action on $X$ as $L_h:X\rightarrow X,x\mapsto g\cdot x$. Then $TX$ is an $H$-space via the induced map $(x,\xi)\mapsto (g\cdot x,L_{g\ast}\xi)$ and from this we can define an $H$-action on $T^\ast X$ and $\Lambda^nT^\ast X$ in the canonical way. For the complex to be invariant it suffices to show that $L_h^\ast$ commutes with $d$, $\star$ and $\iota_v$. But $d$ always commutes with the induced map and it is easy to see that, by virtue of $v$ being the induced vector field for the $H$-action whose elements are isometries, also $\star$ and $\iota_v$ commute with $L_h^\ast$.} and also $G$-invariant. We would now like to compute the index of this complex, $\iind D$ which we embark on in the next section.
        
        At the end of this section we have to issue a warning. It is suggested above that the equivariant index of the complex \eqref{eq--complex} already suffices to compute the superdeterminant. However, the quadratic piece $\Q V^{(2)}$ might have some
        zero-modes that have to be taken care of before performing the path integral. These might arise from the ghost fields as well as from the scalars when expanding around a reducible connection. The ghost zero-modes can be removed in a systematic fashion by adding ghosts of ghosts and pairs $(a_i,\Q a_i)$ of constant fields (they can be thought of as an extension of the BRST content to the non-minimal sector); see \cite{arXiv:0712.2824} for a detailed exposition on $S^4$.  Consequently, the complex \eqref{eq--complex} has to be extended (trivially) by the constant fields $a_i$ which give an additional contribution to $\iind D$ (which is just some integer, corresponding to the number of pairs introduced, times the character of the adjoint; see the next section for details). It is only after we have taken care of these zero-modes that we can translate the index into a well-defined superdeterminant.

    \subsection{Transversally Elliptic Complex and the Equivariant Index}

        In contrast to the elliptic complex obtained for the case of the topologically twisted theory \cite{Witten:1988ze}, the complex \eqref{eq--complex} turns out to only be transversally elliptic, which we show in due course. Note that, at the fixed points of $v$, the complex indeed splits into the (folded) SD/ASD complex and the scalar Laplacian as
        \begin{equation}\label{eq--ASD+laplacian}
            \Bigl(
                \begin{tikzcd}[column sep=small]
                    0\ar[r] & \Omega^1_X(\mathfrak{g}_P)\ar[r,"\mathrm{d}^\pm\oplus\mathrm{d}^\dagger"] & \Omega_X^{2\pm}(\mathfrak{g}_P)\oplus\Omega^0_X(\mathfrak{g}_P)\ar[r] & 0
                \end{tikzcd}
            \Bigr)\oplus
            \Bigl(
                \begin{tikzcd}[column sep=small]
                    0\ar[r] & \Omega^0_X(\mathfrak{g}_P)\ar[r,"\Delta"] & \Omega^0_X(\mathfrak{g}_P)\ar[r] & 0        
                \end{tikzcd}
            \Bigr).
        \end{equation}
        which are both elliptic. For ease of notation, let us henceforth denote the SD, ASD complex by $(\Omega^\bullet,\mathrm{d}^+)$, respectively $(\Omega^\bullet,\mathrm{d}^-)$.
        
        In order to evaluate the one-loop superdeterminant we need to compute the equivariant index of \eqref{eq--complex}. 
        This requires knowledge about the symbol $\sigma(\eth\otimes1)$. For the remainder of these notes let us take the view on the symbol as being (a representative of) an element in the equivariant $K$-group $K_{H\times G}$ over the cotangent bundle $\pi:TX\rightarrow X$ of $X$ (we henceforth identify tangent and cotangent bundles via the metric on $X$). A brief review of this viewpoint and some $K$-theory essentials is provided in \autoref{app--ktheory}. 
        This will be useful since in the transversally elliptic case too there exists an index theorem stating that \cite{Atiyah:1974}
        \begin{equation*}
            \iind(\eth\otimes1)=\ind_{H\times G}[\sigma(\eth\otimes 1)],
        \end{equation*}
        for an $R(H\times G)$-module homomorphism\footnote{Note that the $R(H\times G)$-module structure on $K_{h\times G}(TX)$ is induced by the projection of $X$ onto a point and the fact that, for a point, $K_{H\times G}(\{\mathrm{pt}\})\simeq R(H\times G)$. On $\mathcal{D}^\prime(H\times G)$, a representation in $R(H\times G)$ acts by multiplication with its character.}
        \begin{equation}\label{eq--index}
            \ind_{H\times G}:K_{H\times G}(TX)\rightarrow\mathcal{D}^\prime(H\times G),
        \end{equation}
        where $R(H\times G)$ denotes the representation ring of $H\times G$\footnote{The representation ring $R(H\times G)$ of $H\times G$ is obtained by applying the Grothendieck construction to the semigroup of finite-dimensional complex representation spaces of $H\times G$. Multiplication is given by the tensor product.} and $\mathcal{D}^\prime(H\times G)$ denotes the space of distributions over the space of test functions $\mathcal{D}(H\times G)$ on $H\times G$. The map $\ind$ is called the topological index (as opposed to the analytical one). 
        
        Before we inspect $\sigma(\eth\otimes1)$, let us make the following useful observation. It was briefly discussed in the last section that the equivariant index can still be defined as
        \begin{equation*}
            \iind(\eth\otimes1)=\chara\ker (\eth\otimes1)-\chara\coker (\eth\otimes1).
        \end{equation*}
        But $\eth\otimes 1$ acts trivially on $\Gamma(\mathfrak{g}_P)$ in \eqref{eq--complex} and so does $H$ (remember that we expand around the trivial connection), whereas $G$ acts in the adjoint representation. Therefore, the index can be simplified to
        \begin{equation}\label{eq--fullindex}
            \iind(\eth\otimes1)=\iind\eth\cdot\chi_\mathrm{Ad}
        \end{equation}
        with $\chi_\mathrm{Ad}$ the character of the adjoint representation of $G$. Hence, it suffices to compute the index of
        \begin{equation}\label{eq--complexH}
           \begin{tikzcd}[column sep=scriptsize]
                0\ar[r] & \Omega^1_X\oplus\Omega^0_X\ar[r,"\eth"] & P^+_\omega\Omega^2_X\oplus\Omega^0_X\oplus\Omega^0_X\ar[r] & 0
            \end{tikzcd}
        \end{equation}
        (note that $G$ acts trivially on this complex, hence we expect $\iind\eth\in\mathcal{D}^\prime(H)$).
        
        Let us therefore inspect the symbol $\sigma(\eth)$ of \eqref{eq--complexH} more closely. Using the shorthand notation $\Lambda^i_X=\Lambda^iT^\ast X$, we have the following complex:
        \begin{equation}\label{eq--symbol}
            \begin{tikzcd}[column sep=scriptsize]
                0\ar[r] & \pi^\ast\left(\Lambda^1_X\oplus\Lambda^0_X\right)\ar[r,"\sigma(\eth)"] & \pi^\ast\left(P^+_\omega\Lambda^2_X\oplus\Lambda^0_X\oplus\Lambda^0_X\right)\ar[r] & 0,
            \end{tikzcd}
        \end{equation}
        such that the restriction to the fiber over $(x,\xi)\in X\times T_xX$ gives a linear map
        \begin{equation}\label{eq--symbolcf}
                \sigma(\eth)(x,\xi):(a,\varphi)\longmapsto\left(P^+_\omega[\xi\wedge a+\star(\xi\wedge\kappa)\varphi],\|\xi\|^2\iota_va-\xi_v\langle\xi,a\rangle-\|\xi\|^2c_\omega\varphi,-\langle\xi,a\rangle\right)
        \end{equation}
        with $a\in\Lambda^1_X|_x$ and $\varphi\in\Lambda^0_X|_x$.
        Here, we have introduced the notation $c_\omega:=\cos\omega$ (we also use $s_\omega:=\sin\omega$ later on) and we abuse notation by denoting $\xi$ as an element in $T_xX$ and $T^\ast_xX$ interchangeably. We also have defined $\xi_v:=\langle\xi,v\rangle$, where $\langle\cdot,\cdot\rangle$ denotes the metric inner product in the appropriate sense. 
        For a proper choice of local coordinates and basis sections, \eqref{eq--symbolcf} coincides with the matrix presentation (149) in \cite{arXiv:1812.06473} obtained from a five-dimensional setting. 
        
        It turns out that the Laplacian part of the symbol can be decoupled globally rather than just at the fixed points. The new symbol, while different to the one of \eqref{eq--ASD+laplacian}, has the same support and is ultimately used for the index computation. The decoupling is achieved by applying two maps of complexes, $f=(f^0,f^1)$ and $g=(g^0,g^1)$, to \eqref{eq--symbol}. The maps
        \begin{align*}
            f^0:\pi^\ast(\Lambda^1_X\oplus\Lambda^0_X)&\longrightarrow\pi^\ast(\Lambda^1_X\oplus\Lambda^0_X)\\
            f^1:\pi^\ast(P^+_\omega\Lambda^2_X\oplus\Lambda^0_X\oplus\Lambda^0_X)&\longrightarrow\pi^\ast(P^+_\omega\Lambda^2_X\oplus\Lambda^0_X\oplus\Lambda^0_X)
        \end{align*}
        are bundle morphisms, such that on each fiber we have the linear maps
        \begin{equation}\label{eq--chainmap}
            \begin{aligned}
                f^0:(a,\varphi)&\longmapsto\left(\left(1-P_v\right)a+c_\omega P_v a+\kappa\varphi,-c_\omega\varphi+\iota_va\right),\\
                f^1:(\chi,\tilde{c},c)&\longmapsto(\chi,\tilde{c}-\xi_vc,c),
            \end{aligned}
        \end{equation}
        where we have introduced a projection in $v$-direction, 
        \begin{equation*}
            P_v:=\frac{\kappa\wedge\iota_v}{\|v\|^2}.
        \end{equation*}
        Note that $f^0,f^1$ are isomorphisms. The new symbol map, denoted by $f(\sigma(\eth))$, is given by $f_2\circ\sigma(\eth)\circ f_1^{-1}$. The second set of maps,
        \begin{align*}
            g^0:\pi^\ast(\Lambda^1_X\oplus\Lambda^0_X)&\longrightarrow\pi^\ast(\Lambda^1_X\oplus\Lambda^0_X)\\
            g^1:\pi^\ast(P^+_\omega\Lambda^2_X\oplus\Lambda^0_X\oplus\Lambda^0_X)&\longrightarrow\pi^\ast(P^+_\omega\Lambda^2_X\oplus\Lambda^0_X\oplus\Lambda^0_X)
        \end{align*}
        are the bundle morphisms that, on the fiber over $(x,\xi)\in X\times T_xX$, are linear maps
        \begin{equation}\label{eq--chainmap2}
            \begin{aligned}
                g^0:(a,\varphi)&\longmapsto(a,\varphi),\\
                g^1:(\chi,\tilde{c},c)&\longmapsto\left(\chi,\tilde{c},c+\frac{\xi_v}{\|\xi\|^2}\tilde{c}\right).
            \end{aligned}
        \end{equation}
        Note that $g^1$ is well-defined when applied to $f(\sigma(\eth))$. Also $g^0,g^1$ are isomorphisms and we thus find that $\sigma(\eth)$ is isomorphic to the new symbol complex
        \begin{equation}\label{eq--Symbol}
            \begin{tikzcd}[column sep=scriptsize]
                0\ar[r] & \pi^\ast\left(\Lambda^1_X\oplus\Lambda^0_X\right)\ar[r,"\sigma"] & \pi^\ast\left(P^+_\omega\Lambda^2_X\oplus\Lambda^0_X\oplus\Lambda^0_X\right)\ar[r] & 0,
            \end{tikzcd}
        \end{equation}
        with $\sigma=g\circ f(\sigma(\eth))$,
        acting on the fiber over $(x,\xi)$ as
        \begin{equation}\label{eq--symbolaction}
            \sigma(x,\xi):(a,\varphi)\longmapsto\left(P^+_\omega\left[\xi\wedge\left(1-P_v\right)a+\star\left(\xi\wedge P_v a\right)\right],\|\xi\|^2\varphi,-\langle\xi,\left(1-P_v\right)a+c_\omega P_va\rangle\right).
        \end{equation}
        Note that, on the level of fibers and upon an explicit choice of bases, the symbol is just a matrix and $f,g$ simply act as basis transformations on that matrix. Hence, we can view \eqref{eq--symbolaction} as a ``basis transformation'' of \eqref{eq--symbolcf}.
        
        It is evident from \eqref{eq--symbolaction} that $\sigma=\sigma_\omega+\sigma(\Delta)$ globally, where $\sigma(\Delta):\pi^\ast\Lambda^0_X\rightarrow\pi^\ast\Lambda^0_X$ is the symbol of the Laplacian $\Delta$ and $\sigma_\omega$ maps between the residual summands in \eqref{eq--Symbol} with the action on the fibers specified by \eqref{eq--symbolaction}. By virtue of the isomorphism between $\sigma(\eth)$ and $\sigma$, we can henceforth choose to work with $\sigma$. In particular, at the fixed points of $v$, $\sigma_\omega$ is isomorphic to the symbol $\sigma(\mathrm{d}^\pm)$ of the (folded) SD/ASD complex. 
        
        \begin{prop}
            The symbol complex \eqref{eq--Symbol} is a transversally elliptic complex.
        \end{prop}
        \begin{proof}
            Let $C:=\{x\in X|c_\omega=0\}$. Note that for $x\in C$, $\xi\in T_xX$ such that $\xi$ points in the direction of $v$, the linear map over the fiber simplifies to
            \begin{equation*}
                \sigma(x,\xi):(a,\varphi)\longmapsto(0,\|\xi\|^2\varphi,0).
            \end{equation*}
            This linear map clearly is not invertible. Thus, the symbol map cannot be an isomorphism along $v$ over $C$ and $\sigma$ cannot be elliptic. However, note that the $\sigma(\Delta)$-summand of the symbol is indeed elliptic (recall that $X$ is compact), so the failure in ellipticity of $\sigma$ can be attributed completely to $\sigma_\omega$. Now let us consider the restriction of $TX$ to the subbundle transversal to the $H$-action\footnote{Depending on whether or not the orbit of the Lie algebra generator of $v$ is closed, the isometry group is effectively $S^1$ or the full $T^2$. Both cases are denoted by $H$.},
            \begin{equation*}
                T_HX=\{V\in TX|\forall u\in\Lie H:\langle V,(u^\#)_{\pi(V)}\rangle=0\}\subset TX.
            \end{equation*}
            Here, $u^\#$ denotes the fundamental vector field corresponding to $u$. We claim that
            \begin{equation*}
                \sigma_\omega|_{T_HX}:a\longmapsto\left(P^+_\omega\left[\xi\wedge (1-P_v)a+\star(\xi\wedge P_va)\right],-\langle\xi,a\rangle\right)
            \end{equation*}
            is an elliptic symbol, i.e. the map above is an isomorphism outside the zero-section in $T_HX$:
            \begin{itemize}
                \item Injectivity: Let $\sigma_\omega(x,\xi)a=0$ for $x\in X$ and $\xi\in T_{H,x}X\backslash\{0\}$ arbitrary but fixed. This gives 
                \begin{equation*}
                    P^+_\omega[\xi\wedge(1-P_v)a+\star(\xi\wedge P_va)]=0,\qquad\langle\xi,a\rangle=0.
                \end{equation*}
                Using the explicit form \eqref{eq--projector} of $P^+_\omega$ for the first equation gives
                \begin{equation*}
                    0=\xi\wedge\left((1-P_v)a+c_\omega P_va\right)+\star\left(\xi\wedge\left(c_\omega(1-P_v)a+P_va\right)\right).
                \end{equation*}
                We use this to write
                \begin{align*}
                    0&=\xi\wedge\left((1-P_v)a+c_\omega P_va\right)\wedge\xi\wedge\left(c_\omega(1-P_v)a+P_va\right)\\
                    &=\big(\xi\wedge((1-P_v)a+c_\omega P_va)\big)\wedge\star\big(\xi\wedge((1-P_v)a+c_\omega P_va)\big)\\
                    &=\|\xi\wedge((1-P_v)a+c_\omega P_va)\|^2
                \end{align*}
                which implies $(1-P_v)a+c_\omega\frac{\iota_va}{\|v\|^2}\kappa =\xi\wedge h$ for some zero-form $h$. We now apply $\langle\xi,a\rangle=0$, using that $\xi$ and $\kappa$ are orthogonal, which implies $h=0$. Thus,
                \begin{equation*}
                    (1-P_v)a+c_\omega P_va=0\quad\Longrightarrow\quad (1-P_v)a=0,\,P_va=0.
                \end{equation*}
                But then $a=0$.
                \item Surjectivity: Let $(\chi,b)\in(\pi|_{T_HX})^\ast(P^+_\omega\Lambda^2_X\oplus\Lambda^0_X)$. For $x\in X$ and $\xi\in T_{H,x}X\backslash\{0\}$ arbitrary but fixed, choose\footnote{In order to see that $\sigma_\omega|_{T_HX}(x,\xi)a=(\chi,b)$ we use $P_v\iota_\xi=\iota_\xi P_v$ (remember $\langle\xi,\kappa\rangle=0$) as well as $\star P_v\star B=(1-P_v)B$ and $\iota_\xi(\xi\wedge\star B)=\star(\xi\wedge\iota_\xi B)$ for any two-form $B$.}
                \begin{equation*}
                    a=\frac{1}{\|\xi\|^2}\big((1+c_\omega^2)(1-P_v)\iota_\xi\chi+2c_\omega P_v\iota_\xi\chi+s_\omega^2P_v\iota_\xi\star\chi-b\xi\big).
                \end{equation*}
            \end{itemize}
        \end{proof}
        As a corollary of the proposition above we have that $[\sigma|_{T_HX}]\in K_H(T_HX)$, by virtue of $X$ being compact. Moreover, due to $(g\circ f)$ being an isomorphism we have $[\sigma(\eth)|_{T_HX}]=[\sigma|_{T_HX}]$. Thus, we can compute the equivariant index of \eqref{eq--complexH} with respect to $H$ via the topological index \eqref{eq--index},
        \begin{equation}\label{eq--defindex}
            \iind\eth=\ind_H[\sigma|_{T_HX}]\in\mathcal{D}^\prime(H).
        \end{equation}
        We henceforth only consider the symbol $\sigma|_{T_HX}$ and the index $\ind_H$ and therefore drop the subscripts in the following.
        The remaining task is now to compute the right-hand side of \eqref{eq--defindex}.

\section{Trivialization of the Symbol Complex}\label{sec--2}
    
    Explicit cohomological formulas for the equivariant index of transversally elliptic operators have been introduced by Berline, Vergne \cite{Berline:1996} and Paradan, Vergne \cite{Paradan:2009}. Although these formulas could in principle be employed to compute \eqref{eq--defindex}, they involve an integral of equivariant characteristic classes over a non-compact space whose computation is, in general, quite involved. The work of this section therefore follows the original, $K$-theoretic treatment of the index computation by Atiyah \cite{Atiyah:1974} which proves to be easier in our case. 

    The idea of the subsequent procedure is to split $K_H(T_HX)$ into smaller spaces over which we have good control. More precisely, we have seen in the last section that, at the fixed points, the complex (given by a sum of SD/ASD and Laplacian) is actually elliptic. Therefore, it is desirable to express the full symbol in terms of the ones at just the fixed points. This can be done by constructing a new symbol which reduces to the original one at, and whose support (i.e. ``the part that contributes to the index'') reduces to the fixed points. However, this has to be done such that the new symbol is homotopic to the old one, which guarantees that their index agrees (cf. \cite{Atiyah:1974} Theorem 2.6). The construction of this new symbol is the objective of this section.
     
    First, note that we have $[\sigma]=[\sigma_\omega]+[\sigma(\Delta)]$ from the previous section. Since the index of the Laplacian is easy to compute straight away, we ignore it for now and focus on $[\sigma_\omega]$. Note that the index homomorphism acts on $K_H(T_HX)$, while our elliptic complexes live over the (isolated) fixed points, the set of which is denoted by $Y$ in the following. Therefore, we need to find a map that extends those elliptic complexes, in a consistent manner (see below), over all of $T_HX$. Luckily, \cite{Atiyah:1974} provides us with the existence of such a map:

    For a generic $H$-action on $X$ we get a decreasing filtration
    \begin{equation}\label{eq--filtration}
        X=X_0\supset X_1\supset X_2\supset X_3=\emptyset
    \end{equation}
    with $X_i:=\{x\in X|\dim H_x\ge i\}$ and $H_x$ being the stabilizer of $x$. The sets $X_i-X_{i+1}=\{x\in X|\dim H_x=i\}$ are finite unions of locally closed submanifolds of $X$. Specifically, $X_1$ is the submanifold of fixed points of an $S^1$-subgroup of $H$ whereas $X_2=:Y$ is the set of isolated torus fixed points. For the filtration \eqref{eq--filtration} it was shown in \cite{Atiyah:1974} that there exist homomorphisms $\theta_i$ and split short exact sequences
    \begin{equation}\label{eq--SES}
        \begin{tikzcd}[column sep=small]
            0\ar[r] & K_H(T_H(X-X_i))\ar[r] & K_H(T_H(X-X_{i+1}))\ar[r,shift right] & \ar[l,"\theta_i"',shift right] K_H(T_HX|_{X_i-X_{i+1}})\ar[r] & 0
        \end{tikzcd}
    \end{equation}
    that can be used to, recursively, arrive at the decomposition
    \begin{equation}\label{eq--decomposition}
        K_H(T_HX)=\bigoplus_{i=0}^2\theta_i K_H(T_HX|_{X_i-X_{i+1}}).
        \end{equation}
    Hence, by virtue of the filtration, we can ``break up'' the symbol class $[\sigma_\omega]\in K_H(T_HX)$ into simpler pieces living in $K_H(T_HX|_{X_i-X_{i+1}})$ via some maps $\theta_i$. However, doing this for a generic element in $K_H(T_HX)$ requires knowledge about all levels, not just the top one, $K_H(T_HX|_Y)$ (note that $T_HX|_Y=TX|_Y$).
    In order to be able to construct the new symbol homotopic to $\sigma_\omega$ we have to make the following assumption, which we will justify in due course:
    \begin{assum}\label{as--main}
        There exists some global vector field induced by the group action that can be used to trivialize the symbol $\sigma_\omega$ everywhere outside of (the zero-section over) $Y$.
    \end{assum}
    \noindent Here, what we mean by ``trivialising'' $\sigma_\omega$ is to find the new symbol of said properties. As it turns out, the desired vector field can be identified as the Killing vector field $v$ emerging from the superalgebra\footnote{More precisely, in the following we use the vector field that enters in the definition of $P^+_\omega$ which does not necessarily have to be $v$. However, both choices yield isomorphic subbundles of $\Omega^2$ \cite{arXiv:1812.06473}.}. 
    \autoref{as--main} implies that we obtain $[\sigma_\omega]\in K_H(T_HX)$ as
    \begin{equation}\label{eq--assumption}
        [\sigma_\omega]=[0]+\theta_2[\sigma_\omega|_Y]\in K_H(T_H(X-Y))\oplus \theta_2K_H(TX|_Y),
    \end{equation}
    using \eqref{eq--SES} for $i=2$ (for if there was a non-trivial contribution from $K_H(T_H(X-Y))$, the resulting  symbol would have support also outside of $Y$). Thus, given $\sigma_\omega|_Y$, $[\sigma_\omega]$ is entirely determined by the homomorphism $\theta_2$ which we construct momentarily.
    
    \begin{const}\label{con--theta2}
        We want $\theta_2$ to extend the symbol class $[\sigma_\omega|_Y]$ at the torus fixed points to a class in $K_H(T_HX)$. At $TX|_Y$, the symbol is smooth and elliptic and therefore can be extended to $TU$ over some (possibly small) open neighborhood $U\supset Y$ while preserving ellipticity\footnote{Use a retraction $r:U\rightarrow Y$ which induces $r^\ast:K_H(TY)\rightarrow K_H(TU)$.}. Note that, since $U$ is open, the zero-section over $U$ is not compact and thus the extension of $\sigma_\omega|_Y$ restricted to $T_HU$ does not have compact support. This can be remedied by pushing the support away from the zero-section on $U-Y$ along the vector field $v$ (which, by \autoref{as--main} only vanishes at $Y$) such that, on $T_HU$, the support reduces to $Y$ which is again compact. We can perform the push via two maps
        \begin{equation}\label{eq--deformationmap}
            f^\pm:T_HU\rightarrow TU,\,\,(x,\xi)\mapsto(x,\xi\pm g(|\xi|)v(x))
        \end{equation}
        depending on whether we push in the direction of $v$ or against it. Here, $(x,\xi)\in U\times T_{H,x}U$ are local coordinates and $g$ can be taken to be a bump function
        \begin{align*}
            &g:\mathbb{R}_{\ge0}\rightarrow[0,1],\,u\mapsto
            \begin{cases}
                \exp(1-\frac{1}{1-u^2}), & 0\le u<1\\
                0, & u\ge1    
            \end{cases}.
        \end{align*}  
        The role of $g(|\xi|)$ is to make sure that, far away from the zero-section, \eqref{eq--deformationmap} returns to the identity.
        Schematically, the deformation can be depicted as follows:
        
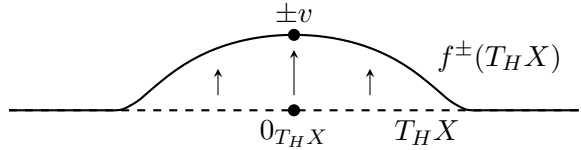
\begin{center}
            \begin{tikzpicture}[x=2.5cm]
                \draw[dashed,thick] (-1.5,0)--(1.5,0);
                \fill (0,0) circle [radius=.08cm] node[black] [below]{$0_{T_HX}$};
                \draw[thick,domain=-.99:.99,smooth,variable=\t] plot (\t,{exp(1-1/(1-\t*\t))});
                \draw[thick] (-1.5,0)--(-.99,0);
                \draw[thick] (.99,0)--(1.5,0);
				\draw[-stealth] (0,.2)--(0,.8);
				\draw[-stealth] (-.4,.2)--(-.4,.55);
				\draw[-stealth] (.4,.2)--(.4,.55);
				\fill (0,1) circle [radius=.08cm] node [above]{$\pm v$};
				\node at (.7,0) [below]{$T_HX$};
				\node at (.7,.7) [right]{$f^\pm(T_HX)$};
            \end{tikzpicture}
        \end{center}
        We call $f^\pm$ the deformation maps; they induce homomorphisms
        \begin{equation*}
            (f^\pm)^\ast:K_H(TU)\rightarrow K_H(T_HU),\,\,[\sigma]\mapsto[\sigma\circ f^\pm].
        \end{equation*} 
        The resulting deformed symbol can then be extended to an element in $K_H(T_HX)$ using the natural extension homomorphism\footnote{The natural extension homomorphism for an open inclusion $\iota:U\rightarrow X$ is obtained as the induced homomorphism from the map $X^+\rightarrow X^+/(X^+-U^+)\simeq U^+$, where the $+$ superscript denotes one-point compactification of the space (in particular, $X^+=X\cup{\mathrm{pt}}$ for compact spaces $X$). Intuitively, this map projects all the stuff living outside of $U$ to just a point, since we only care about what happens on $U$. For the induced map, this translates to a ``trivial extension''.} $\iota_!$ for the open inclusion $\iota:T_HU\hookrightarrow T_HX$.
    \end{const} 

    We summarise the construction above in the following definition.
    
    \begin{defi}(Extension homomorphism $\theta_2$)
        Let $X$ be a compact, simply-connected smooth manifold with a $T^2$-action producing a discrete set $Y$ of torus fixed points. Let $U$ be an open neighborhood around $Y$ and $r:U\rightarrow Y$ a retraction. We define the \textit{extension homomorphism} 
        \begin{equation}\label{eq--exthom}
            \begin{tikzcd}[column sep=small]
                \theta_2: K_H(TX|_Y)\ar[r,"r^\ast"] & K_H(TU)\ar[r,"f^\ast"] & K_H(T_HU)\ar[r,"\iota_!"] & K_H(T_HX)
            \end{tikzcd}
        \end{equation}
        with $f^\ast$ induced by \eqref{eq--deformationmap}.
    \end{defi}
    
    \begin{remark}
        The construction of $\theta_2$ is independent of the choice of retraction $r$ and also of the choice of open neighborhood $U$ (so long as the extension to $K_H(TU)$ is still elliptic).    
    \end{remark}
    
    \begin{eg}\label{ex--c2}
        Let us illustrate the idea of the construction on a simple example\footnote{Within the example, for clarity, we restore the proper notation for cotangent spaces.}. Consider $X=\mathbb{C}^2$ (i.e. we are only looking at one patch) and the complex
        \begin{equation*}
            \bar{\partial}:\Omega^{0,0}\rightarrow\Omega^{0,1}\rightarrow\Omega^{0,2}.
        \end{equation*}
        The vector field $v$ is taken to be
        \begin{equation*}
            v=\mathrm{i}\epsilon_1(z\partial_z-\bar{z}\partial_{\bar{z}})+\mathrm{i}\epsilon_2(w\partial_w-\bar{w}\partial_{\bar{w}})
        \end{equation*}
        with $\epsilon_1,\epsilon_2\neq0$ (NB: if $\epsilon_1/\epsilon_2\not\in\mathbb{Q}$ then the orbit of $v$ is dense in $T^2$) and one fixed point at the origin of $\mathbb{C}^2$ ($Y=\{0\}$). Since $\Omega^{0,0}$ is just the space of section of a trivial complex line bundle over $\mathbb{C}^2$ and $\Omega^{0,1},\Omega^{0,2}$ the spaces of section in $T^\ast_{0,1}\mathbb{C}^2,\Lambda^2T^\ast_{0,1}\mathbb{C}^2$, the symbol is given by
        \begin{equation*}
            \sigma:\pi^\ast(\mathbb{C}^2\times\mathbb{C})\rightarrow \pi^\ast (T^{\ast}_{0,1}\mathbb{C}^2)\rightarrow\pi^\ast (\Lambda^2T^{\ast}_{0,1}\mathbb{C}^2),
        \end{equation*}
        where $\pi:T^\ast\mathbb{C}^2\rightarrow\mathbb{C}^2$. On elements in the fibre over $((z,w),\xi)\in T^\ast\mathbb{C}^2$ the symbol acts as 
        \begin{equation}\label{eq--exsymb}
            \sigma((z,w),\xi)\,u=\xi^{0,1}\wedge u.
        \end{equation}
        Note that $\sigma$ is already elliptic on $\mathbb{C}^2$ (it is an isomorphism for $\xi^{0,1}\neq0$). Therefore, starting on $Y$, the extension to an open neighborhood is simply \eqref{eq--exsymb}. Next, in order to reduce the support of \eqref{eq--exsymb} to $Y$ we apply $f^\ast$ which, as stated above, acts by precomposition of \eqref{eq--deformationmap}:
        \begin{align}\label{eq--exsymbdef}
            f^\ast\sigma((z,w),\xi_T)&=\sigma((z,w),\xi_T\pm g(|\xi_T|)\kappa)\nonumber\\&=\xi_T^{0,1}\mp\begin{cases}
                \exp(1-\frac{1}{1-|\xi_T|^2})\mathrm{i}(\epsilon_1z\mathrm{d}\bar{z}+\epsilon_2w\mathrm{d}\bar{w}), & |\xi_T|\in[0,1)\\
                0, & \text{else}
            \end{cases},
        \end{align}
        where $\xi_T$ denotes elements transversal to $\kappa=g(v,\cdot)$, i.e. in the radial directions of $\mathbb{C}^2$. Therefore, the new symbol \eqref{eq--exsymbdef} is invertible everywhere (including the zero-section $\xi_T=0$), except when $z=w=0$, i.e. except at the fixed point. 
        The role of $g(|\xi_T|)$ is to make sure that, far away from the zero-section ($|\xi_T|>1$), \eqref{eq--exsymbdef} returns to \eqref{eq--exsymb} (restricted to $\xi_T$).
        It is easy to see that \eqref{eq--exsymbdef} is a continuous deformation of \eqref{eq--exsymb} (restricted to $\xi_T$).
        Finally, in this simple setup, we can compute the elements in the cohomology of, say, the $f^+$-deformed operator $\bar{\partial}-\mathrm{i}(\epsilon_1z\mathrm{d}\bar{z}+\epsilon_2w\mathrm{d}\bar{w})$ explicitly:
        \begin{equation}\label{eq--excoh}
            \begin{aligned}
                H^0:&\quad\{z^nw^m\e{\mathrm{i}\epsilon_1|z|^2+\mathrm{i}\epsilon_2|w|^2}\}_{n,m\in\mathbb{N}_{0}},\\
                H^1:&\quad\{\bar{z}^nw^m\e{\mathrm{i}\epsilon_1|z|^2+\mathrm{i}\epsilon_2|w|^2}\mathrm{d}\bar{z}, z^n\bar{w}^m\e{\mathrm{i}\epsilon_1|z|^2+\mathrm{i}\epsilon_2|w|^2}\mathrm{d}\bar{w}\}_{n,m\in\mathbb{N}_0},\\
                H^2:&\quad\{\bar{z}^n\bar{w}^m\e{\mathrm{i}\epsilon_1|z|^2+\mathrm{i}\epsilon_2|w|^2}\mathrm{d}\bar{z}\,\mathrm{d}\bar{w}\}_{n,m\in\mathbb{N}_0}.
            \end{aligned}
        \end{equation}
        Here, for simplicity, we have neglected $g$ in the deformation.
    \end{eg}

    In the definition for $\theta_2$ we were deliberately vague about which of the two maps in \eqref{eq--deformationmap} induce $f^\ast$. The reason for this is the following: Since $Y$ is a discrete set, $U$ a priori is a disjoint union $\bigsqcup_{j\in I}U_j$ of open neighborhoods $U_j$ around each fixed point ($I=\{1,\dots,|Y|\}$) so, in particular, we could use $f^+$ around some fixed points and $f^-$ around the remaining ones. However, if we want to satisfy \autoref{as--main} using our vector field $v$ to push (which is globally defined), it turns out that there are only two valid choices of deformations, which turn out to be equivalent on the level of the index. We make the following

    \begin{claim}\label{cl--claim}
        Given the symbol class $[\sigma_\omega|_Y]\in K_H(TX|_Y)$, after a choice of $f^\pm$ on a neighborhood $U_1$ of one fixed point, there is a unique assignment of deformation maps for all other components of $U=\bigsqcup_{j\in I}U_j$ which either matches the distribution of SD/ASD complexes at the fixed points or its opposite.
    \end{claim}

    This can be seen as follows: Let $F_1,F_2\in Y$ be two fixed points and $U_1,U_2$ the open neighborhoods around $F_1,F_2$, respectively. Moreover, let $\cos\omega=1$ at $F_1$ and $\cos\omega=-1$ at $F_2$, i.e. the symbol $\sigma_\omega$ is isomorphic to the SD one at $F_1$ and to the ASD one at $F_2$. Concretely, from \eqref{eq--symbolaction} we obtain
    \begin{subequations}\label{eq--FPsymbols}
        \begin{align}
            \sigma_\omega|_{F_1}(x,\xi):a&\longmapsto (P^+[\xi\wedge(1-P_v)a+\star(\xi\wedge P_va)],-\langle\xi,(1-P_v)a+P_va\rangle),\\
            \sigma_\omega|_{F_2}(x,\xi):a&\longmapsto (P^-[\xi\wedge(1-P_v)a+\star(\xi\wedge P_va)],-\langle\xi,(1-P_v)a-P_va\rangle)\label{eq--symbolSD}
        \end{align}
    \end{subequations}
    as the representatives
    of the respective symbol class $[\sigma_\omega|_{F_{1,2}}]$.

    Since the symbols \eqref{eq--FPsymbols} are in fact elliptic on all of $TX$ (remember that, at the fixed points, they are isomorphic to SD/ASD), we can extend $U_1,U_2$ to open sets $V_1\supset U_1,V_2\supset U_2$ such that $V_1\cap V_2\neq\emptyset$ (in fact, since $X$ is compact, we can extend the neighborhoods of all the fixed points to an open cover of $X$). Applying the natural extension homomorphism in \eqref{eq--exthom} should yield a symbol homotopic to $\sigma_\omega$ and such that \eqref{eq--assumption} holds; this is only true if the deformation of the symbols indeed extends in a compatible way on the intersection. It is immediate from \eqref{eq--FPsymbols} that on $V_1\cap V_2$, we have the following equality\footnote{In order for this equality to be meaningful, the first summand of one of the two symbols has to be mapped from $\Omega^{2+}$ to $\Omega^{2-}$ or vice versa. The construction of $P^+_\omega$ (cf. \cite{arXiv:1812.06473}) dictates to use $-m$ in \eqref{eq--transition} for this map.}:  
    \begin{equation}\label{eq--cc}
        \sigma_\omega|_{F_1}(x,v(x))=\sigma_\omega|_{F_2}(x,-v(x)).
    \end{equation}
    Hence, in order to be able to trivialize the symbol everywhere on $T(V_1\cup V_2)$ we need to deform along the direction of $v$ on $V_1$ and against it on $V_2$, i.e. we use $f^+$ on $T_HU_1$ and $f^-$ on $T_HU_2$. For if we were to use, say, $f^+$ on both, the resulting symbol would no longer be a continuous deformation, unless the push in $v$-direction vanishes on some subset of $V_1\cap V_2$. But then the resulting symbol would be supported on the zero-section over that subset which contradicts \autoref{as--main}.
    
    If at $F_2$ the symbol was instead isomorphic to the SD one, the deformations would be trivially compatible on $V_1\cap V_2$. 
    Clearly, by checking compatibility according to \eqref{eq--cc} for all fixed points, one finds that the assignment of $f^\pm$ is dictated precisely by the distribution of SD/ASD complexes at the fixed points. Note that we equally could have chosen to deform with $f^-$ on $T_HU_1$ and $f^+$ on $T_HU_2$. For example, on some $X$ with three fixed points and complexes $(F_1,\text{SD}),(F_2,\text{ASD}),(F_3,\text{ASD})$, we can use either $f^+,f^-,f^-$ or $f^-,f^+,f^+$, respectively. However, we will find in \autoref{sec--3} that both choices yield the same index, i.e. the ambiguity gets resolved on the level of the index.

    \begin{remark}\label{rem--homotopy}
        Let us recapitulate here what our argument was. If \autoref{as--main} holds true, then knowledge of $\theta_2$ is enough to give a description of the symbol $[\sigma_\omega]$ in $K_H(T_HX)$ entirely in terms of its description at the fixed points. But our construction above shows that, using the Killing vector field $v$ provided by the supersymmetry background, we can always construct a $\theta_2$ such that the assumption is satisfied: this is ensured by choosing $\theta_2$ such that \eqref{eq--cc} holds, which makes sure that the support of the symbol is pushed off the zero-section in $v$-direction everywhere on $X$ except the fixed points. 
        
        Moreover, it is easy to see that $\theta_2\sigma_\omega|_Y$ is a continuous deformation of $\sigma_\omega$ (it is continuously deformed along $v$ in a neighborhood of the zero-section in $T_HU_i$ and the condition \eqref{eq--cc} ensures that this deformation extends continuously on $T_HX$), i.e. the symbols are homotopic and thus, $[\sigma_\omega]=\theta_2[\sigma_\omega|_Y]\in K_H(T_HX)$.
        
        Finally, note that the construction of $\theta_2$ only depends on the form of the symbol at the fixed points and not at any intermediate points in the manifold. 
    \end{remark}

    \begin{remark}
        Note that, while our choice of $v$ and, correspondingly, \autoref{con--theta2} indeed satisfies \autoref{as--main} for the symbol of \eqref{eq--complex}, this is not, in general, the case for other symbols. In order to determine their decomposition in \eqref{eq--decomposition} one would have to construct the maps $\theta_i$ for $i<2$ corresponding to the lower levels of the filtration \eqref{eq--filtration} too.
    \end{remark}

\section{Index Computation}\label{sec--3}

    Before we determine the index of $\sigma$ for the complex \eqref{eq--Symbol} using the machinery developed in \cite{Atiyah:1974}, let us consider the following
    
    \begin{eg}
        We return to the setting in \autoref{ex--c2} were we computed explicitly the kernel and cokernel of the deformed operator. Suppose now that $\epsilon_1,\epsilon_2$ have a small imaginary part and suppose further that $\IM\epsilon_1,\IM\epsilon_2>0$. Then, only the exponentially decaying modes are allowed, i.e. 
        \begin{equation*}
            H^0=\{z^nw^m\e{\mathrm{i}\epsilon_1|z|^2+\mathrm{i}\epsilon_2|w|^2}\}_{n,m\in\mathbb{N}_0},\qquad H^1=\emptyset,\qquad H^2=\emptyset
        \end{equation*}
        (since either $\epsilon_1$ or $\epsilon_2$ turns into $\bar{\epsilon}_1$, respectively $\bar{\epsilon}_2$ for $H^1$ and for $H^2$ both of them do). Hence, denoting the basic character of the two $S^1$ by $t_1$, $t_2$, the equivariant index in this case gives
        \begin{equation*}
            \iind=\sum_{n,m\in\mathbb{N}_0}t_1^nt_2^m=\sum_{n\in\mathbb{N}_0}t_1^n\sum_{m\in\mathbb{N}_0}t_2^m=\left(\frac{1}{1-t_1}\right)^+\left(\frac{1}{1-t_2}\right)^+
        \end{equation*}
        where the notation in the last equality is explained below and can be ignored for now. Had we deformed the complex using $f^-$, all exponentials in \eqref{eq--excoh} would have a minus sign in the exponent and $H^2$ becomes the relevant one, hence
        \begin{equation*}
            \iind=\sum_{n,m\in\mathbb{N}_0}t_1^{-(n+1)}t_2^{-(m+1)}=\sum_{n\in\mathbb{N}}t_1^{-n}\sum_{m\in\mathbb{N}}t_2^{-m}=\left(\frac{1}{1-t_1}\right)^-\left(\frac{1}{1-t_2}\right)^-.
        \end{equation*}
        Although this is the basic idea of the index computation, the complex \eqref{eq--complex} is much more complicated, so we resort to more abstract methods below. 
    \end{eg}
    
    We have seen in the last section that $[\sigma]=[\sigma_\omega]+[\sigma(\Delta)]$ and $\sigma_\omega$ can be replaced by a symbol $\theta_2\sigma_\omega|_Y$ which is supported only at the fixed points and homotopic to $\sigma_\omega$. Now we can apply the index homomorphism,
    \begin{equation}\label{eq--indexsum}
        \iind\eth=\ind\theta_2[\sigma_\omega|_Y]+\ind[\sigma(\Delta)].
    \end{equation}
    First, let us take care of the Laplacian part. Its index vanishes trivially by virtue of $\Delta$ being self-adjoint on $\Omega^\bullet_X$.
    
    Now for the first part in \eqref{eq--indexsum}. Although we have come a long way from reducing the original symbol \eqref{eq--symbol} to essentially $\sigma_\omega|_Y$, it turns out that we can ``break it up'' even further. For this, note that $TX|_Y$ can be viewed as a complex vector bundle over $TY$, namely as its normal bundle. Hence, there is a Thom isomorphism $\phi:K_H(TY)\overset{\simeq}{\longrightarrow} K_H(TX|_Y)$. It was shown in \cite{Atiyah:1974} that this Thom isomorphism acts as multiplication\footnote{There is a product $K_H(TY)\otimes K_{H\times U(2)}(T_H\mathbb{C}^2)\rightarrow K_H(TX|_Y)$, with $\mathbb{C}^2$ the fiber of the normal bundle. See theorem 4.3 in \cite{Atiyah:1974} for details.} by $[\sigma(\bar{\partial})]\in K_{H\times U(2)}(T_H\mathbb{C}^2)$. Hence, our element $[\sigma_\omega|_Y]\in K_H(TX|_Y)$ can be written as a (tensor) product of $[\sigma(\bar{\partial})]$ and some element in $K_H(TY)$ (to be determined below).
    
    Finally, since the symbol $\theta_2[\sigma_\omega|_Y]$ is trivialized everywhere except in neighborhoods of the fixed points, we can apply the excision property in \cite{Atiyah:1974} (theorem 3.7) to yield a sum over local contributions to the index,
    \begin{equation}\label{eq--indexomega}
        \ind\theta_2[\sigma_\omega|_Y]=-\sum_{i\in I}\ind\theta^{s_i}_2[\sigma(\bar{\partial})]\cdot\ind[\sigma_{F_i}].
    \end{equation}
    Here, $s_i\in\{-,+\}$ and $\theta_2^{s_i}$ denotes the extension homomorphism \eqref{eq--exthom} restricted to $U_i$ (i.e. only going the first two steps in \eqref{eq--exthom}), using as deformation map $f^{s_i}$. Note that the product of symbols via Thom isomorphism is respected by the index\footnote{See, for example, theorem 3.5 in \cite{Atiyah:1974}.}. Lastly, $[\sigma_{F_i}]\in K_H(T\{F_i\})\simeq R(H)$ denotes the element such that when multiplied with $[\sigma(\bar{\partial})]$ we obtain $[\sigma_\omega|_{F_i}]\in K_H(TX|_{F_i})$. Note that the symbols on the right in \eqref{eq--indexomega} are over $U_i\simeq\mathbb{R}^4\simeq\mathbb{C}^2$.

    In order to determine the symbol class $[\sigma_{F_i}]$ we use the following propositions, whose proofs are sketched in \autoref{app--isomorphism}:
    
    \begin{prop}\label{prop--iso1}
        For the (complexified) SD complex $(\Omega^\bullet,\mathrm{d}^+)_\mathbb{C}$ on $\mathbb{C}^2$ there is an isomorphism
        \begin{equation}
            (\Omega^\bullet,\mathrm{d}^+)_\mathbb{C}\simeq(\Omega^{0,\bullet}\otimes(\mathcal{O}\oplus\Lambda^{2,0}T^\ast\mathbb{C}^2),\bar{\partial}\otimes 1).
        \end{equation}
    \end{prop}

    \begin{prop}\label{prop--iso2}
        For the SD and ASD complexes, $(\Omega^\bullet,\mathrm{d}^+)$ and $(\Omega^\bullet,\mathrm{d}^-)$, on $\mathbb{C}^2$ there is an isomorphism
        \begin{equation}
            (\Omega^\bullet,\mathrm{d}^+)\simeq(\Omega^\bullet,\mathrm{d}^-),
        \end{equation}
        induced by the map $\mathbb{C}^2\ni(z_1,z_2)\mapsto(\bar{z}_1,z_2)$.
    \end{prop}
    \noindent From \autoref{prop--iso1} we can directly read off $\sigma_{F_i}$ for $\sigma_\omega|_{F_i}\simeq\sigma(\mathrm{d}^+)$ as the complex of length zero given by\footnote{Since $\mathcal{O}\oplus\Lambda^{2,0}T^\ast\mathbb{C}^2$ is just a vector space, we can simply view it as a bundle over $F_i$ or $T\{F_i\}$, hence its class is in $K_H(T\{F_i\})$.}
    \begin{equation}
        (\mathcal{O}\oplus\Lambda^{2,0}T^\ast\mathbb{C}^2):
        \begin{tikzcd}[column sep=small]
            \dots\ar[r] & 0\ar[r] & 0\ar[r] & \mathcal{O}\oplus\Lambda^{2,0}T^\ast\mathbb{C}^2\ar[r] & 0\ar[r] & 0\ar[r] & \dots
        \end{tikzcd}
    \end{equation} 
    with $\mathcal{O}\oplus\Lambda^{2,0}T^\ast\mathbb{C}^2$ at level zero. Note that the minus sign in \eqref{eq--indexomega} is due to the fact that the symbol at the fixed point is isomorphic to the one of the folded SD/ASD complex (cf. \eqref{eq--ASD+laplacian}) which starts at level one (while the $\bar{\partial}$-complex starts at level zero). 
    
    For $\sigma_\omega|_{F_i}\simeq\sigma(\mathrm{d}^-)$ at a given fixed point one simply applies \autoref{prop--iso2} first. This has the effect of flipping the first weight for the $T^2$-action around $F_i$. Note that, equally well, we could have chosen the isomorphism of \autoref{prop--iso2} to be induced by $(z_1,z_2)\mapsto(z_1,\bar{z}_2)$ instead. One can check that the index \eqref{eq--eqindex} is still the same. 
    
    Finally, we are in a position to explicitly compute the full index of \eqref{eq--complex}. Consider a patch $U_l\simeq\mathbb{C}^2$ around some fixed point $F_l\in Y$ with complex coordinates $(z_1,z_2)$. Let $\epsilon_1,\epsilon_2$ be coordinates on $\Lie G$ such that $t_{1}=\exp(\mathrm{i}\epsilon_{1}),t_2=\exp(\mathrm{i}\epsilon_{2})$ are coordinates on $G$. We can then express the infinitesimal weights $\alpha^{(l)}_1,\alpha^{(l)}_2$ for the group action\footnote{Note that for fixed points with ASD complexes one needs to take the change in (local) complex structure according to \autoref{prop--iso2} into account when determining $\alpha^{(l)}_{ij}$.} on $(z_1,z_2)\in\mathbb{C}^2$ as $\alpha^{(l)}_i=\sum_{j=1}^2\alpha^{(l)}_{ij}\epsilon_j$ with $\alpha^{(l)}_{ij}\in\mathbb{Z}$ for $i,j=1,2$. Then we get
    \begin{equation}
        \ind[\sigma_{F_l}]=1+\prod_{i=1}^2t^{-\alpha^{(l)}_i},\qquad t^{\alpha^{(l)}_i}:=\prod_{j=1}^2t_j^{\alpha^{(l)}_{ij}}.
    \end{equation}
    The first factor in \eqref{eq--indexomega} was computed in \cite{Atiyah:1974} Theorem 8.1 and is given by
    \begin{equation}\label{eq--laurent}
        \ind\theta^{\pm}_2[\sigma(\bar{\partial})]=\prod_{i=1}^2\left(\frac{1}{1-t^{-\alpha^{(l)}_i}}\right)^\pm
    \end{equation}
    with $(\cdot)^\pm$ denoting the Laurent expansion around $t=0$ and $t=\infty$, respectively\footnote{The process of expanding around $t=0$ or $t=\infty$ is commonly referred to as regularization.}. Finally, assembling all individual contributions from above, the complex \eqref{eq--complex} has
    \begin{equation}\label{eq--eqindex}
        \iind\eth=-\sum_{l\in I}\left(1+\prod_{i=1}^2t^{-\alpha^{(l)}_i}\right)\prod_{k=1}^2\left(\frac{1}{1-t^{-\alpha^{(l)}_k}}\right)^{s_l}.
    \end{equation}
    The full equivariant index, taking into account the gauge-part, is obtained as \eqref{eq--fullindex}. This is our main result. To summarize, by noticing that \autoref{as--main} holds for the complex \eqref{eq--complex} obtained from localization, we are able to globally push the original symbol $\sigma$ off the zero-section outside of the fixed points using $\theta_2$, thereby reducing its support to $Y$. We then employ the filtration with respect to the group action to break down the index of the symbol in simpler pieces that we can finally evaluate explicitly.
    
    \begin{remark}
        (Initial choice of deformation map) We found in \autoref{sec--2} an ambiguity in the construction of $\theta_2$ arising from the initial choice of deformation map for the first fixed point. However, both choices, $f^+$ and $f^-$, lead to the same index which follows immediately from the fact that the resulting deformed symbols are both homotopic to the original one (see \autoref{rem--homotopy}). Hence, their index must agree. In the examples considered below, we show explicitly that the ambiguity gets resolved on the level of the index.
    \end{remark}
    
    \begin{remark}
        (Non-trivial gauge backgrounds) In cases where the complex arises from the expansion around topologically non-trivial connections such that $c_1\neq0$, i.e. there is flux on $X$, the index is no longer given simply by \eqref{eq--fullindex}. This is because there seems to be no canonical way of defining an $H$-action on the gauge bundle. It is therefore appealing to only consider $H$-equivariant bundles which was proposed in the five-dimensional setup \cite{arXiv:1904.12782} by introducing equivariant curvature and fluxes. It was shown that this leads to a shift of the Coulomb parameter $a_0\mapsto a_0+k_i(\epsilon_1^i,\epsilon_2^i)$, where $k_i$ is a function of the flux and isometry parameters $\epsilon_1,\epsilon_2$ at each fixed point $i$. Consequently, we would have to replace \eqref{eq--fullindex} by a sum over the fixed points of the individual contributions in \eqref{eq--eqindex} multiplied by $\chi_\mathrm{Ad}(a_0+k_i(\epsilon_1^i,\epsilon_2^i))$. However, it is an open problem at this point to formulate the index for non-zero flux for arbitrary $X$ and SD/ASD distributions.
        
        In any case, the isometry part at each fixed point can still be obtained as the summands in \eqref{eq--eqindex}, i.e. the general procedure (and, in particular, the regularization) can be extended to the non-zero flux case.
    \end{remark}
    
    \begin{remark}
        (Hypermultiplet) Although not demonstrated explicitly here, the same rule for the choice of deformation maps should hold for the hypermultiplet symbol, i.e. choose $f^+$ at SD fixed points and $f^-$ at ASD ones (or vice versa). For the hypermultiplet, the complex at the fixed points roughly equates to the one of a chiral/anti-chiral Dirac operator (cf. \cite{Festuccia:2020yff}). Hence, upon applying appropriate transition maps (similar to \eqref{eq--transition}), the two symbols, once extended to large enough open subsets, should be related in terms of \eqref{eq--cc} on the intersection. On $S^4$ this was shown in \cite{arXiv:0712.2824,Hama:2012bg}.
    \end{remark}

    Hence, when working with localization on an $\mathcal{N}=2$ theory over a simply-connected compact four-manifold $X$, knowledge of the $H=T^2$-action (i.e. the Killing vector field) and the distribution of SD/ASD over its fixed points immediately gives the index as \eqref{eq--eqindex} and thereby the one-loop partition function. In particular, we showed that the index can be computed from the local (elliptic) contributions around the fixed points for an arbitrary assignment of SD/ASD and these local contributions are combined by applying the correct Laurent expansions as determined in \autoref{sec--2}.

\section{Examples}\label{sec--4}
    In this section we apply our result \eqref{eq--eqindex} for the index to various examples of four-manifolds with $H=T^2$-action, with different distributions of SD/ASD complexes at the fixed points of the $T^2$-action. For convenience, we state here the Laurent expansions used in \eqref{eq--eqindex} explicitly:
    \begin{equation}
    \begin{aligned}\label{eq--laurent.t-1}
        \left(\frac{1}{1-t^{-\alpha_i}}\right)^+&=-\sum_{n=1}^\infty t^{n\alpha_i}=-t^{\alpha_i}-t^{2\alpha_i}-\dots\\
        \left(\frac{1}{1-t^{-\alpha_i}}\right)^-&=\sum_{n=0}^\infty t^{-n\alpha_i}=1+t^{-\alpha_i}+t^{-2\alpha_i}+\dots\\
        \left(\frac{1}{1-t^{\alpha_i}}\right)^+&=\sum_{n=0}^\infty t^{n\alpha_i}=1+t^{\alpha_i}+t^{2\alpha_i}+\dots\\
        \left(\frac{1}{1-t^{\alpha_i}}\right)^-&=-\sum_{n=1}^\infty t^{-n\alpha_i}=-t^{-\alpha_i}-t^{-2\alpha_i}-\dots
    \end{aligned}
    \end{equation}
    Most of the examples considered below have been presented in \cite{arXiv:1812.06473,arXiv:1904.12782} where they are obtained from  five-dimensional considerations. Our results, obtained in a purely four-dimensional way, can be matched exactly by applying the following dictionary between expansions:
    \begin{equation}\begin{split}\label{eq--laurent.conv}
        &\left(\frac{1}{1-t^{-\alpha_i}}\right)^+\longleftrightarrow\left[\frac{1}{1-t^{-\alpha_i}}\right]^-,\\[.6em]
        &\left(\frac{1}{1-t^{-\alpha_i}}\right)^-\longleftrightarrow\left[\frac{1}{1-t^{-\alpha_i}}\right]^+,
    \end{split}\end{equation}
    where $[\cdot]^\pm$ denotes the expansions used in \cite{arXiv:1812.06473}. Some examples presented below cannot be obtained from five dimensions and are new results. At the end of this section we show how to compute from the index the perturbative part of the SYM partition function (in the zero-flux sector) on $\mathbb{F}^1$.

    \subsection{Sphere $S^4$} 
    
        We describe $S^4$ as the quaternion projective space $\mathbb{HP}^1$ with elements $[q_1,q_2]\sim[q_1q,q_2q]$ for $q\in\mathbb{H}^\times$ and introduce local inhomogeneous coordinates $q=q_1q_2^{-1}=z_1+jz_2$ on the northern patch, where $z_1,z_2\in\mathbb{C}$, and $q^{-1}=(\bar{z}_1-jz_2)/|q|^2$ on the southern patch. This choice gives local complex coordinates $(z_1,z_2)$ on the northern patch and $(\bar{z}_1,-z_2)$ on the southern patch. $T^2$ acts by left-multiplication, $q_1\mapsto t_1q_1$ and $q_2\mapsto t_2q_2$, which yields $z_1\mapsto t_1t_2^{-1}z_1$ and $z_2\mapsto t_1^{-1}t_2^{-1}z_2$. From this action we can read off
        \begin{equation*}
        (\alpha_{ij})=\begin{pmatrix}
        1 & -1 \\
        -1 & -1 
        \end{pmatrix}.
        \end{equation*}
        Note that the action of $T^2$ on $\bar{z}_1$ is instead $\bar{z}_1\mapsto t_1^{-1}t_2\bar{z}_1$ and $\alpha_{ij}$ changes correspondingly.
        
        We consider two complexes $(E_1^\bullet,\eth_1)$ and $(E_2^\bullet,\eth_2)$ over $S^4$ given by \eqref{eq--complexH} for different SD/ASD distributions. In $(E_1^\bullet,\eth_1)$ we place SD complexes at both poles whereas for $(E_2^\bullet,\eth_2)$ we place a SD complex at the north pole and a ASD complex at the south pole. On the southern patch, for $(E_2^\bullet,\eth_2)$, we employ \autoref{prop--iso2} and consider the isomorphic SD complex. The isomorphism is induced by the map $\bar{z}_1\mapsto z_1$, $-z_2\mapsto -z_2$ which implies a flip of the weight of $z_1$ on the southern patch. Note that the weights which flip in this way are highlighted in boldface in all the examples below. The SD/ASD distribution and weights around the fixed points for the two complexes can be conveniently displayed by the ``Delzant polygon'' of $S^4$:
        
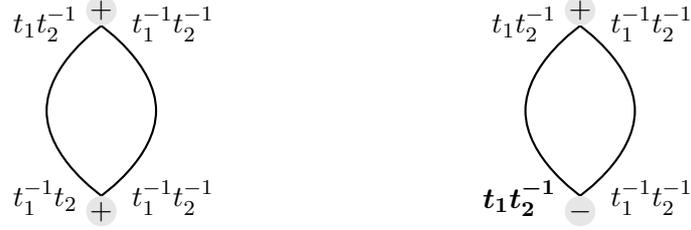
\begin{center}
            \begin{tikzpicture}[scale=.9]
                \def\y{1.25}
                \def\x{.8}
                \def\s{7}
                \def\f{20}
                \draw[thick] plot [smooth,tension=1] coordinates {(0,0) (\x,\y) (0,2*\y)};
                \draw[thick] plot [smooth,tension=1] coordinates {(0,0) (-\x,\y) (0,2*\y)};
                \draw[thick] plot [smooth,tension=1] coordinates {(\s,0) (\s+\x,\y) (\s,2*\y)};
                \draw[thick] plot [smooth,tension=1] coordinates {(\s,0) (-\x+
                \s,\y) (\s,2*\y)};
                \node[fill=gray!\f,circle,inner sep=0pt,minimum size=4mm] at (0,0) [below]{$+$};
                \node[fill=gray!\f,circle,inner sep=0pt,minimum size=4mm] at (0,2*\y) [above]{$+$};
                \node at (-.2,.3*\y) [below left]{$t_1^{-1}t_2$};
                \node at (.3,.3*\y) [below right]{$t_1^{-1}t_2^{-1}$};
                \node at (-.2,2*\y-.3*\y) [above left]{$t_1t_2^{-1}$};
                \node at (.3,2*\y-.3*\y) [above right]{$t_1^{-1}t_2^{-1}$};
                \node[fill=gray!\f,circle,inner sep=0pt,minimum size=4mm] at (\s,0) [below]{$-$};
                \node[fill=gray!\f,circle,inner sep=0pt,minimum size=4mm] at (\s,2*\y) [above]{$+$};
                \node at (-.2+\s,.3*\y) [below left]{$\boldsymbol{t_1t_2^{-1}}$};
                \node at (.3+\s,.3*\y) [below right]{$t_1^{-1}t_2^{-1}$};
                \node at (-.2+\s,2*\y-.3*\y) [above left]{$t_1t_2^{-1}$};
                \node at (.3+\s,2*\y-.3*\y) [above right]{$t_1^{-1}t_2^{-1}$};
            \end{tikzpicture}
        \end{center}
        The choice of deformation maps follows from \autoref{cl--claim} and is completely determined by the distribution of SD/ASD complexes. Hence it is ``$+$'' at both poles for $(E_1^\bullet,\eth_1)$ and ``$+$'' at the north and ``$-$'' at the south pole for $(E_2^\bullet,\eth_2)$.
        
        Applying \eqref{eq--eqindex} yields the index of the two complexes:
        \begin{equation}
        \begin{aligned}\label{eq--asdindex}
            \iind\eth_1=&\,-(1+t_2^{-2})\left(\frac{1}{1-t_1t_2^{-1}}\right)^+\left(\frac{1}{1-t_1^{-1}t_2^{-1}}\right)^+\\& -(1+t_1^{-2}) \left(\frac{1}{1-t_1^{-1}t_2}\right)^+\left(\frac{1}{1-t_1^{-1}t_2^{-1}}\right)^+,
        \end{aligned}
        \end{equation}
        \vspace{.4em}
        \begin{equation}\label{eq--flipindex}
        \begin{aligned}
            \iind\eth_2=&\,-(1+t_2^{-2})\left(\frac{1}{1-t_1 t_2^{-1}}\right)^+\left(\frac{1}{1-t_1^{-1}t_2^{-1}}\right)^+\\& -(1+\boldsymbol{t_2^{-2}}) \left(\frac{1}{1-\boldsymbol{t_1t_2^{-1}}}\right)^-\left(\frac{1}{1-t_1^{-1}t_2^{-1}}\right)^-.
        \end{aligned}
        \end{equation}
        Since $(E^\bullet_1,\eth_1)$ is everywhere a SD complex, in particular, it is elliptic. Hence, we expect the index to be an element of $R(H)$. This is confirmed by our computation:
        \begin{equation}
            \iind\eth_1=-1.
        \end{equation}
        The complex $(E^\bullet_2,\eth_2)$ on the contrary is transversally elliptic, hence its equivariant index will be an infinite power series in $(t_1 t_2^{-1})$ and $(t_1^{-1}t_2^{-1})$ with each term appearing with a finite multiplicity. In \autoref{figure1} we plot the exponents $n_1$ and $n_2$ appearing for $(t_1t_2^{-1})$ and $(t_1^{-1}t_2^{-1})$, respectively and the corresponding multiplicity for each term.
        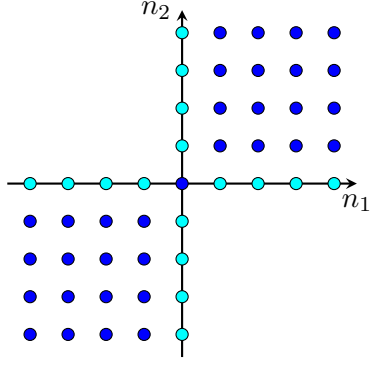
\begin{figure}[h!]
        \centering
        \begin{tikzpicture}
            \def\rad{.08cm};
            \draw[thick,-stealth] (-2.3,0)--(2.3,0) node [below]{$n_1$};
            \draw[thick,-stealth] (0,-2.3)--(0,2.3) node [left]{$n_2$};
            \foreach \x in {.5,1,...,2}{
                \foreach \y in {.5,1,...,2}{
                    \draw[fill=blue] (\x,\y) circle [radius=\rad];
                    \draw[fill=blue] (-\x,-\y) circle [radius=\rad];
                }
            };
            \foreach \n in {.5,1,...,2}{
                \draw[fill=light] (0,\n) circle [radius=\rad];
                \draw[fill=light] (\n,0) circle [radius=\rad];
                \draw[fill=light] (0,-\n) circle [radius=\rad];
                \draw[fill=light] (-\n,0) circle [radius=\rad];
            }
            \draw[fill=blue] (0,0) circle [radius=\rad];
        \end{tikzpicture}
        \caption{We show the exponents of the weights in \eqref{eq--flipindex} for $S^4$. Light blue points have multiplicity one, blue points have multiplicity two.}
        \label{figure1}
        \end{figure}
       
        In order to exemplify the comparison with \cite{arXiv:1812.06473}, we consider the north pole contribution of $\iind\eth_1$. The corresponding contribution in \cite{arXiv:1812.06473} is given by:
        \begin{equation}
            \left[\frac{1}{1-t_1t_2^{-1}}\right]^+\left[\frac{1}{1-t_1^{-1}t_2^{-1}}\right]^+ +\left[\frac{1}{1-t_1^{-1}t_2}\right]^-\left[\frac{1}{1-t_1 t_2}\right]^-.
        \end{equation}
        Upon applying \eqref{eq--laurent.conv} we find:
        \begin{equation}\begin{split}
            &\left(\frac{1}{1-t_1t_2^{-1}}\right)^+\left(\frac{1}{1-t_1^{-1}t_2^{-1}}\right)^+ +\left(\frac{1}{1-t_1^{-1}t_2}\right)^+\left(\frac{1}{1-t_1 t_2}\right)^+\\
            =&\left(\frac{1}{1-t_1t_2^{-1}}\right)^+\left(\frac{1}{1-t_1^{-1}t_2^{-1}}\right)^+ +\left(\frac{t_1t_2^{-1}}{1-t_1t_2^{-1}}\right)^+\left(\frac{t_1^{-1}t_2^{-1}}{1-t_1^{-1}t_2^{-1}}\right)^+\\
            =&(1+t_2^{-2})\left(\frac{1}{1-t_1t_2^{-1}}\right)^+\left(\frac{1}{1-t_1^{-1}t_2^{-1}}\right)^+
        \end{split}\end{equation}
        which matches \eqref{eq--asdindex}. All other contributions can be matched in the same way. 
        
        In \autoref{sec--3} we commented on how the ambiguity for the initial choice of deformation map (here, the choice of ``$+$'' or ``$-$'' at the north pole, from which ``$-$'' or ``$+$'' then follows for the south pole) gets resolved on the level of the index. We can now check this explicitly by taking the other choice for the deformation map at the north pole which
        yields:
        \begin{equation}
        \begin{aligned}
            \iind\eth_1=&-(1+t_2^{-2})\left(\frac{1}{1-t_1t_2^{-1}}\right)^-\left(\frac{1}{1-t_1^{-1}t_2^{-1}}\right)^-\\& -(1+t_1^{-2}) \left(\frac{1}{1-t_1^{-1}t_2}\right)^-\left(\frac{1}{1-t_1^{-1}t_2^{-1}}\right)^-,
        \end{aligned}
        \end{equation}
        \vspace{.4em}
        \begin{equation}
        \begin{aligned}
            \iind\eth_2=&-(1+t_2^{-2})\left(\frac{1}{1-t_1t_2^{-1}}\right)^-\left(\frac{1}{1-t_1^{-1}t_2^{-1}}\right)^-\\& -(1+\boldsymbol{t_2^{-2}}) \left(\frac{1}{1-\boldsymbol{t_1t_2^{-1}}}\right)^+\left(\frac{1}{1-t_1^{-1}t_2^{-1}}\right)^+.
        \end{aligned}
        \end{equation}  
        One can compare the contribution from each fixed point to that in \eqref{eq--asdindex}, \eqref{eq--flipindex}. The effect of changing the initial choice of deformation map, for each fixed point contribution, is to reverse the exponents in the power series: 
        \begin{equation}\label{eq--symmetry}
            (n_1,n_2)\longmapsto (-n_1,-n_2).
        \end{equation}
        However, it is apparent from \autoref{figure1} that the index is centrally symmetric in the $(n_1,n_2)$-plane and therefore it is independent of the initial choice of deformation map. We will see below that this point symmetry around the origin is also present for the other examples, hence, there is no ambiguity on the level of the index.
        
    \subsection{Complex Projective Space $\mathbb{CP}^2$}
    
        On $\mathbb{CP}^2$ there is a $T^2$-action with three fixed points present. We consider patches $U_i=\{[z_1,z_2,z_3]|z_i\neq0\}$ with $i=1,2,3$ around each fixed point and introduce on $U_1$ inhomogeneous coordinates $(z_2/z_1,z_3/z_1)$, on whom $T^2$ acts as $z_2/z_1\mapsto t_1 z_2/z_1$, $z_3/z_1\mapsto t_2z_3/z_1$. This produces
        \begin{equation*}
        (\alpha_{ij})=\begin{pmatrix}
        1 & \phantom{.}0 \\
        0 & \phantom{.}1 
        \end{pmatrix}.
        \end{equation*}
        Similarly we cover $U_2,U_3$ with inhomogeneous coordinates $(z_1/z_2,z_3/z_2)$ and $(z_2/z_3,z_1/z_3)$ and obtain the respective $\alpha_{ij}$. We consider two complexes $(E_1^\bullet,\eth_1)$ and $(E_2^\bullet,\eth_2)$ whose distribution of SD/ASD and the weights of the $T^2$-action at the fixed points are given as follows\footnote{Note that rotating the $+/-$ distribution does not affect the final result for the index. This can be viewed as just a relabeling of the patches.}:     
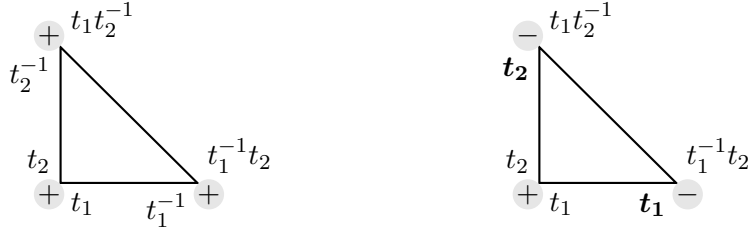
        \begin{center}
            \begin{tikzpicture}[scale=.9]
                \def\x{2}
                \def\y{2}
                \def\s{5}
                \def\f{20}
                \draw[thick] (0,0) node[fill=gray!\f,circle,inner sep=0pt,minimum size=4mm] [below left]{$+$} --(0,\y) node[fill=gray!\f,circle,inner sep=0pt,minimum size=4mm] [above left]{$+$} --(\x,0) node[fill=gray!\f,circle,inner sep=0pt,minimum size=4mm] [below right]{$+$} --cycle;
                \draw[thick] (\s+\x,0) node[fill=gray!\f,circle,inner sep=0pt,minimum size=4mm] [below left]{$+$} --(\s+\x,\y) node[fill=gray!\f,circle,inner sep=0pt,minimum size=4mm] [above left]{$-$} --(\s+\x+\x,0) node[fill=gray!\f,circle,inner sep=0pt,minimum size=4mm] [below right]{$-$} --cycle;
                \node at (0,0) [below right]{$t_1$};
                \node at (0,0) [above left]{$t_2$};
                \node at (0,\y) [below left]{$t_2^{-1}$};
                \node at (0,\y) [above right]{$t_1t_2^{-1}$};
                \node at (\x,0) [below left]{$t_1^{-1}$};
                \node at (\x,0) [above right]{$t_1^{-1}t_2$};
                \node at (\s+\x,0) [below right]{$t_1$};
                \node at (\s+\x,0) [above left]{$t_2$};
                \node at (\s+\x,\y) [below left]{$\boldsymbol{t_2}$};
                \node at (\s+\x,\y) [above right]{$t_1t_2^{-1}$};
                \node at (\s+\x+\x,0) [below left]{$\boldsymbol{t_1}$};
                \node at (\s+\x+\x,0) [above right]{$t_1^{-1}t_2$};
            \end{tikzpicture}  
        \end{center}
        In the figure on the right, $t_1$ and $t_2$ flip at the minus fixed points due to the use of \autoref{prop--iso2}. We apply \eqref{eq--eqindex} to obtain the index of the complexes:
        \begin{equation}
        \begin{aligned}
            \iind\eth_1=&-(1+t_1 t_2)\left(\frac{1}{1-t_1}\right)^+\left(\frac{1}{1-t_2}\right)^+ -(1+t_1^{-2}t_2) \left(\frac{1}{1-t_1^{-1}}\right)^+\left(\frac{1}{1-t_1^{-1}t_2}\right)^+\\ &-(1+t_1t_2^{-2})\left(\frac{1}{1-t_1t_2^{-1}}\right)^+\left(\frac{1}{1-t_2^{-1}}\right)^+,
        \end{aligned}
        \end{equation}
        \vspace{.4em}
        \begin{equation}\label{eq--CP2flip}
        \begin{aligned}
            \iind\eth_2=&-(1+t_1t_2)\left(\frac{1}{1-t_1}\right)^+\left(\frac{1}{1-t_2}\right)^+ -(1+\boldsymbol{t_2}) \left(\frac{1}{1-\boldsymbol{t_1}}\right)^-\left(\frac{1}{1-t_1^{-1}t_2}\right)^-\\ &-(1+\boldsymbol{t_1})\left(\frac{1}{1-t_1t_2^{-1}}\right)^-\left(\frac{1}{1-\boldsymbol{t_2}}\right)^-.
        \end{aligned}
        \end{equation}
        The first complex is again associated to an elliptic differential operator on $\mathbb{CP}^2$ and therefore an element in $R(H)$:
        \begin{equation}
            \iind\eth_1=-2.
        \end{equation}
        The complex $(E^\bullet_2,\eth_2)$ is transversally elliptic and thus it is a power series in $t_1$ and $t_2$, with finite multiplicities. The exponents and multiplicities are displayed in the $(n_1,n_2)$-plane in \autoref{figure2}. The resulting plot seems identical to the one on  $S^4$ in \autoref{figure1}, however, note that the weights are different for both cases, hence, the index is too.
        \begin{figure}[h!]
        \centering
        \begin{tikzpicture}
            \def\rad{.08cm};
            \draw[thick,-stealth] (-2.3,0)--(2.3,0) node [below]{$n_1$};
            \draw[thick,-stealth] (0,-2.3)--(0,2.3) node [left]{$n_2$};
            \foreach \x in {.5,1,...,2}{
                \foreach \y in {.5,1,...,2}{
                    \draw[fill=blue] (\x,\y) circle [radius=\rad];
                    \draw[fill=blue] (-\x,-\y) circle [radius=\rad];
                }
            };
            \foreach \n in {.5,1,...,2}{
                \draw[fill=light] (0,\n) circle [radius=\rad];
                \draw[fill=light] (\n,0) circle [radius=\rad];
                \draw[fill=light] (0,-\n) circle [radius=\rad];
                \draw[fill=light] (-\n,0) circle [radius=\rad];
            }
            \draw[fill=blue] (0,0) circle [radius=\rad];
        \end{tikzpicture}
        \caption{We show the exponents of the weights in \eqref{eq--CP2flip} for $\mathbb{CP}^2$. Light blue points have multiplicity one, blue points have multiplicity two.}
        \label{figure2}
        \end{figure}
        The two complexes considered above correspond to the SD and flip' cases in \cite{arXiv:1812.06473} and the index can again be seen to match those results (upon applying \eqref{eq--laurent.conv}). Finally, also in this case $\iind\eth_2$ is symmetric under a reflection $(n_1,n_2)\mapsto(-n_1,-n_2)$ and the final result is independent of the initial choice of deformation map.

    \subsection{Hirzebruch Surface $\mathbb{F}^1$}
    
        As a last example we consider the Hirzebruch surface $\mathbb{F}_1=\{(z_1,z_2;u_1,u_2)\}/\sim$ with $(z_1,z_2),(u_1,u_2)\in\mathbb{C}^2\backslash\{0\}$ and
        \begin{equation*}
            (z_1,z_2;u_1,u_2)\sim(z_1^\prime,z_2^\prime;u_1^\prime,u_2^\prime):\Longleftrightarrow\exists\lambda,\mu\in\mathbb{C}^\times:(z_1^\prime,z_2^\prime;u_1^\prime,u_2^\prime)=(\lambda z_1,\lambda z_2;\lambda\mu u_1,\mu u_2).
        \end{equation*}
        Here we have a $T^2$-action $z_1\mapsto t_1z_1$, $u_1\mapsto t_2u_1$ with four fixed points. The four patches are covered by the usual choice of inhomogeneous coordinates.
        We consider three complexes $(E_1^\bullet,\eth_1)$,  $(E_2^\bullet,\eth_2)$ and $(E_3^\bullet,\eth_3)$ whose SD/ASD distribution and weights are given, respectively, by
        \begin{center}
            \begin{tikzpicture}[scale=.9]
                \def\x{1.5}
                \def\y{1.5}
                \def\dif{1.5}
                \def\s{6}
                \def\f{20}
                \draw[thick] (0,0) node[fill=gray!\f,circle,inner sep=0pt,minimum size=4mm] [below left]{$+$} --(0,\y) node[fill=gray!\f,circle,inner sep=0pt,minimum size=4mm] [above left]{$+$} --(\x,\y+\dif) node[fill=gray!\f,circle,inner sep=0pt,minimum size=4mm] [above right]{$+$} --(\x,0) node[fill=gray!\f,circle,inner sep=0pt,minimum size=4mm] [below right]{$+$} --cycle;
                \draw[thick] (\s,0) node[fill=gray!\f,circle,inner sep=0pt,minimum size=4mm] [below left]{$+$} --(\s,\y) node[fill=gray!\f,circle,inner sep=0pt,minimum size=4mm] [above left]{$+$} --(\s+\x,\y+\dif) node[fill=gray!\f,circle,inner sep=0pt,minimum size=4mm] [above right]{$-$} --(\s+\x,0) node[fill=gray!\f,circle,inner sep=0pt,minimum size=4mm] [below right]{$-$} --cycle;
                \draw[thick] (2*\s,0) node[fill=gray!\f,circle,inner sep=0pt,minimum size=4mm] [below left]{$+$} --(2*\s,\y) node[fill=gray!\f,circle,inner sep=0pt,minimum size=4mm] [above left]{$-$} --(2*\s+\x,\y+\dif) node[fill=gray!\f,circle,inner sep=0pt,minimum size=4mm] [above right]{$+$} --(2*\s+\x,0) node[fill=gray!\f,circle,inner sep=0pt,minimum size=4mm] [below right]{$-$} --cycle;
                \node at (0,0) [below right]{$t_1$};
                \node at (0,0) [above left]{$t_2$};
                \node at (\x,0) [below left]{$t_1^{-1}$};
                \node at (\x,0) [above right]{$t_2$};
                \node at (\x,\y+\dif) [below right]{$t_2^{-1}$};
                \node at (\x,\y+\dif) [above left]{$t_1^{-1}t_2^{-1}$};
                \node at (0,\y) [below left]{$t_2^{-1}$};
                \node at (.2,\y+.7) {$t_1t_2$};
                \node at (\s,0) [below right]{$t_1$};
                \node at (\s,0) [above left]{$t_2$};
                \node at (\s+\x,0) [below left]{$\boldsymbol{t_1}$};
                \node at (\s+\x,0) [above right]{$t_2$};
                \node at (\s+\x,\y+\dif) [below right]{$t_2^{-1}$};
                \node at (\s+\x,\y+\dif) [above left]{$\boldsymbol{t_1t_2}$};
                \node at (\s,\y) [below left]{$t_2^{-1}$};
                \node at (\s+.2,\y+.7) {$t_1t_2$};
                \node at (2*\s,0) [below right]{$t_1$};
                \node at (2*\s,0) [above left]{$t_2$};
                \node at (2*\s+\x,0) [below left]{$\boldsymbol{t_1}$};
                \node at (2*\s+\x,0) [above right]{$t_2$};
                \node at (2*\s+\x,\y+\dif) [below right]{$\boldsymbol{t_2}$};
                \node at (2*\s+\x,\y+\dif) [above left]{$\boldsymbol{t_1t_2}$};
                \node at (2*\s,\y) [below left]{$\boldsymbol{t_2}$};
                \node at (2*\s+.2,\y+.7) {$t_1t_2$};
            \end{tikzpicture}
        \end{center}
        Applying the index formula \eqref{eq--eqindex} yields:
        \begin{equation}
        \begin{aligned}
            \iind\eth_1=&-(1+t_1t_2)\left(\frac{1}{1-t_1}\right)^+\left(\frac{1}{1-t_2}\right)^+ -(1+t_1^{-1}t_2) \left(\frac{1}{1-t_2}\right)^+\left(\frac{1}{1-t_1^{-1}}\right)^+\\ &-(1+t_1^{-1}t_2^{-2})\left(\frac{1}{1-t_1^{-1}t_2^{-1}}\right)^+\left(\frac{1}{1-t_2^{-1}}\right)^+\\& -(1+t_1)\left(\frac{1}{1-t_2^{-1}}\right)^+\left(\frac{1}{1-t_1t_2}\right)^+,
        \end{aligned}
        \end{equation}
        \vspace{.4em}
        \begin{equation}\label{eq--F1flip}
        \begin{aligned}
            \iind\eth_2=&-(1+t_1t_2)\left(\frac{1}{1-t_1}\right)^+\left(\frac{1}{1-t_2}\right)^+ -(1+\boldsymbol{t_1t_2}) \left(\frac{1}{1-t_2}\right)^-\left(\frac{1}{1-\boldsymbol{t_1}}\right)^-\\ &-(1+\boldsymbol{t_1})\left(\frac{1}{1-\boldsymbol{t_1t_2}}\right)^-\left(\frac{1}{1-t_2^{-1}}\right)^-\\
            &-(1+t_1)\left(\frac{1}{1-t_2^{-1}}\right)^+\left(\frac{1}{1-t_1t_2}\right)^+,
        \end{aligned}
        \end{equation}
        \vspace{.4em}
        \begin{equation}\label{eq--F1flip'}
        \begin{aligned}
            \iind\eth_3=&-(1+t_1t_2)\left(\frac{1}{1-t_1}\right)^+\left(\frac{1}{1-t_2}\right)^+ -(1+\boldsymbol{t_1t_2}) \left(\frac{1}{1-t_2}\right)^-\left(\frac{1}{1-\boldsymbol{t_1}}\right)^-\\ &-(1+\boldsymbol{t_1t_2^2})\left(\frac{1}{1-\boldsymbol{t_1t_2}}\right)^+\left(\frac{1}{1-\boldsymbol{t_2}}\right)^+\\
            &-(1+\boldsymbol{t_1t_2^2})\left(\frac{1}{1-\boldsymbol{t_2}}\right)^-\left(\frac{1}{1-t_1t_2}\right)^-.
        \end{aligned}
        \end{equation}
        As expected, for the topologically twisted theory we find an element in $R(H)$:
        \begin{equation}
        \iind\eth_1=-2.
        \end{equation}
        Similar to the previous examples the exponents and multiplicities of the weights are displayed in \autoref{figure3} for the transversally elliptic complexes $(E^\bullet_2,\eth_2)$ and $(E^\bullet_3,\eth_3)$.
        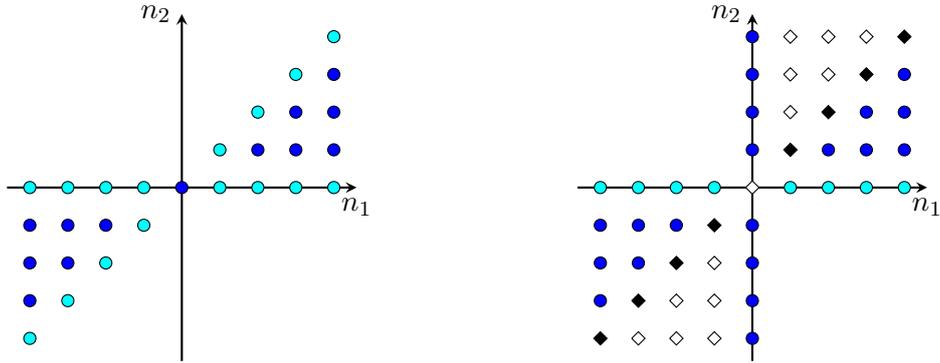
\begin{figure}[h!]
        \centering
        \begin{tikzpicture}
            \def\rad{.08cm};
            \draw[thick,-stealth] (-2.3,0)--(2.3,0) node [below]{$n_1$};
            \draw[thick,-stealth] (0,-2.3)--(0,2.3) node [left]{$n_2$};
            \foreach \x in {1,...,3}{
                \foreach \y in {1,...,\x}{
                    \draw[fill=blue] (\x/2+1/2,\y/2) circle [radius=\rad];
                    \draw[fill=blue] (-\x/2-1/2,-\y/2) circle [radius=\rad];
                }
            };
            \foreach \n in {.5,1,...,2}{
                \draw[fill=light] (\n,\n) circle [radius=\rad];
                \draw[fill=light] (\n,0) circle [radius=\rad];
                \draw[fill=light] (-\n,-\n) circle [radius=\rad];
                \draw[fill=light] (-\n,0) circle [radius=\rad];
            }
            \draw[fill=blue] (0,0) circle [radius=\rad];
        \end{tikzpicture}
        \hspace{6em}
        \begin{tikzpicture}
            \def\rad{.08cm};
            \draw[thick,-stealth] (-2.3,0)--(2.3,0) node [below]{$n_1$};
            \draw[thick,-stealth] (0,-2.3)--(0,2.3) node [left]{$n_2$};
            \foreach \x in {1,...,3}{
                \foreach \y in {1,...,\x}{
                    \draw[fill=blue] (\x/2+1/2,\y/2) circle [radius=\rad];
                    \draw[fill=blue] (-\x/2-1/2,-\y/2) circle [radius=\rad];
                    \draw[fill=white] (\y/2,\x/2+1/2) ++(-2.5pt,0 pt) -- ++(2.5pt,2.5pt)--++(2.5pt,-2.5pt)--++(-2.5pt,-2.5pt)--++(-2.5pt,2.5pt);
                    \draw[fill=white] (-\y/2,-\x/2-1/2) ++(-2.5pt,0 pt) -- ++(2.5pt,2.5pt)--++(2.5pt,-2.5pt)--++(-2.5pt,-2.5pt)--++(-2.5pt,2.5pt);
                }
            };
            \foreach \n in {.5,1,...,2}{
                \draw[fill=black] (\n,\n) ++(-2.5pt,0 pt) -- ++(2.5pt,2.5pt)--++(2.5pt,-2.5pt)--++(-2.5pt,-2.5pt)--++(-2.5pt,2.5pt);
                \draw[fill=light] (\n,0) circle [radius=\rad];
                \draw[fill=black] (-\n,-\n) ++(-2.5pt,0 pt) -- ++(2.5pt,2.5pt)--++(2.5pt,-2.5pt)--++(-2.5pt,-2.5pt)--++(-2.5pt,2.5pt);
                \draw[fill=light] (-\n,0) circle [radius=\rad];
                \draw[fill=blue] (0,\n) circle [radius=\rad];
                \draw[fill=blue] (0,-\n) circle [radius=\rad];
            }
            \draw[fill=white] (0,0) ++(-2.5pt,0 pt) -- ++(2.5pt,2.5pt)--++(2.5pt,-2.5pt)--++(-2.5pt,-2.5pt)--++(-2.5pt,2.5pt);
        \end{tikzpicture}
        \caption{We show the exponents of the weights in \eqref{eq--F1flip} on the left and \eqref{eq--F1flip'} on the right. Light blue points have multiplicity one, blue points multiplicity two, black squares multiplicity three and white squares multiplicity four.}
        \label{figure3}
        \end{figure}
        The result for the index of $(E^\bullet_2,\eth_2)$ matches with \cite{arXiv:1812.06473}, upon using \eqref{eq--laurent.conv}. The index of $(E^\bullet_3,\eth_3)$, on the other hand, cannot be obtained in the five-dimensional framework\footnote{In $d=5$ one usually considers only (contact) instantons. According to \cite{arXiv:1904.12782}, reducing along different fibers produces different distributions of instantons/anti-instantons on the four-dimensional manifold. $(E^\bullet_1,\eth_1)$ and $(E^\bullet_2,\eth_2)$ can be obtained in this way whereas $(E^\bullet_3,\eth_3)$ can not. In some sense, one is ``running out of fibers'' along which to reduce.} but is only accessible through our four-dimensional procedure and is therefore a new result. Finally, also in this case $\iind\eth_2$ and $\iind\eth_3$ are symmetric under a reflection $(n_1,n_2)\mapsto(-n_1,-n_2)$, confirming again that the index is independent of the initial choice of deformation map.
        
        Finally, as a physical application, let us compute the superdeterminant and, hence, the one-loop contribution to the partition function for SYM on $\mathbb{F}^1$ for the ``$+-+-$'' distribution corresponding to $(E^\bullet_3,\eth_3)$. It was explained in \autoref{sec--1} that, in addition to the index result \eqref{eq--F1flip'}, one needs to take possible ghost zero-modes into account. The way to do this is to extend the field content by a number of pairs of constant scalar fields $(a_i,\Q a_i)$, $(\bar{a}_i,\Q\bar{a}_i)$ in the adjoint of the gauge group, corresponding to the amount of zero-modes ($\Q a_i$ has ghost number of $c$, $\Q\bar{a}_i$ has ghost number of $\bar{c}$). Therefore, we have $\Phi=(A,\varphi,a_i,\bar{a}_i)$ now. They appear in $V^{(2)}$ as additional terms of the form $c\wedge\star \bar{a}_i+\bar{c}\wedge\star a_i$. Hence, according to \eqref{eq--V2} we can take them to be zero-modes of $D$ (remember we only care about the terms with highest-order derivatives) and simply add this number of zero-modes to \eqref{eq--F1flip'}.
        
        As was the case throughout this work, we limit our considerations to the zero-flux sector (however, on general grounds one expects the full BPS locus to have flux-carrying solutions as well; $\mathbb{F}^1$ has two independent two-cycles).
        Since we expand around the trivial connection (and $\mathbb{F}^1$ has $H^0(\mathbb{F}^1)=\mathbb{R}$), both $c$ and $\bar{c}$ have one zero-mode, so we have to add $+2$ to our index result \eqref{eq--F1flip'}. For the full index, this has to be multiplied by the character of the adjoint representation of the gauge group (cf. \autoref{sec--1}) and we obtain
        \begin{equation}\label{eq--indexfull}
            (2+\iind\eth)\cdot\chi_\mathrm{Ad},
        \end{equation}
        where
        \begin{equation}\label{eq--chiad}
            \chi_{\mathrm{Ad}}=\mbox{rk}G+\sum_{\alpha\in\Delta}\exp(ia_0\cdot\alpha).
        \end{equation}
        Here, the sum is over the roots of the gauge group and $a_0$ is the Coulomb parameter.
        
        In order to compute the superdeterminant and, thus, the perturbative partition function in terms of the index \eqref{eq--indexfull}, we introduce the $\Upsilon$-function (cf. \cite{arXiv:0712.2824,Hama:2012bg,arXiv:1812.06473}):
        \begin{equation}
            \Upsilon^C(x|\epsilon_1,\epsilon_2)=\prod_{(n_1,n_2)\in C\cap\mathbb{Z}^2}\left(\epsilon_1 n_1+\epsilon_2 n_2+x\right)\prod_{(n_1,n_2)\in C^\circ\cap\mathbb{Z}^2}\left(\epsilon_1 n_1+\epsilon_2 n_2+x\right),
        \end{equation}
        where $C$ denotes a rational cone and $C^\circ$ its interior. One can determine the cones contributing to \eqref{eq--F1flip'} from the plot on the right in \autoref{figure3}. We define two cones $C_1$ and $C_2$ as $C_1:=\{(n_1,n_2)\in\mathbb{Z}^2|\;n_1,n_2\geq 0\}$ and  $C_2:=\{(n_1,n_2)\in\mathbb{Z}^2|\;n_2\geq n_1\geq 0\}$, respectively. Upon converting our index result \eqref{eq--indexfull} to the superdeterminant using \eqref{eq--iindD} and \eqref{eq--sdet}, we find, up to factors independent of $a_0$ (which can be absorbed in the normalization):
        \begin{equation}\begin{split}\label{eq--sdetF1flip'}
            \sdet^{1/2}|_{H^\bullet(D)}\R=\prod_{\alpha\in\Delta^+}\frac{1}{(a_0\cdot\alpha)^2}&\Upsilon^{C_1}(a_0\cdot\alpha|\epsilon_1,\epsilon_2)\Upsilon^{C_1}(-a_0\cdot\alpha|\epsilon_1,\epsilon_2)\\
            \cdot&\Upsilon^{C_2}(a_0\cdot\alpha|\epsilon_1,\epsilon_2)\Upsilon^{C_2}(-a_0\cdot\alpha|\epsilon_1,\epsilon_2).
        \end{split}\end{equation}
        The product is taken over positive roots and the factor in the denominator is due to the ghost contribution and cancels with the Vandermonde determinant arising from the integral over $a_0$ restricted to the Cartan subalgebra \cite{Hama:2012bg}.
        
        Comparing this expression with the result for the superdeterminant in the cases considered in \cite{arXiv:1812.06473} (page 42) one notices that \eqref{eq--sdetF1flip'} has ``double the amount'' of $\Upsilon$-functions. This is strictly related to the fact that the ``$+-+-$'' distribution of complexes on $\mathbb{F}_1$ cannot be obtained by reducing from a five-dimensional Sasakian manifold (i.e. the method used in \cite{arXiv:1812.06473,arXiv:1904.12782}) and represents a novel result. Hence, our index computation constitutes the first step towards computing the full partition function (also including flux, classical and non-perturbative part) of the $\mathcal{N}=2$ SYM theory on any compact, simply-connected manifold with arbitrary SD/ASD distributions.
    
\paragraph*{Acknowledgments}
    
    We are grateful to Jian Qiu for many illuminating discussions and would like to thank him and Guido Festuccia for comments on the manuscript. LR acknowledges support by Vetenskapsrådet under grant 2018-05572. RM acknowledges support by Vetenskapsrådet under grant 2018-05572 and the Centre for Interdisciplinary Mathematics at Uppsala University.

\appendix

    \renewcommand\qedsymbol{}

    \section{Isomorphism of Complexes}\label{app--isomorphism}

        In this section we sketch a proof of the two propositions used in order to compute the index in \autoref{sec--3}.

        \newtheorem*{prop--iso1}{Proposition 4.1}
        \begin{prop--iso1}
            {\it For the (complexified) SD complex $(\Omega^\bullet,\mathrm{d}^+)_\mathbb{C}$ on $\mathbb{C}^2$ there is an isomorphism}
            \begin{equation*}
                (\Omega^\bullet,\mathrm{d}^+)_\mathbb{C}\simeq(\Omega^{0,\bullet}\otimes(\mathcal{O}\oplus\Lambda^{2,0}T^\ast\mathbb{C}^2),\bar{\partial}\otimes 1).
            \end{equation*}
        \end{prop--iso1}
        \begin{proof}(Sketch)
            First, note that the SD complex $(\Omega^\bullet,\mathrm{d}^+)$ is isomorphic \cite{Donaldson:1985} to
            \begin{equation}\label{eq--complexASD}
                \begin{tikzcd}[column sep=scriptsize]
                    0\ar[r] & \Omega^{0}\ar[r] & \Omega^{0,1}\oplus\overline{\Omega^{0,1}}\ar[r] & \Omega^{0,2}\oplus\overline{\Omega^{0,2}}\oplus\Omega^{1,1}_{\parallel}\ar[r] & 0
                \end{tikzcd},
            \end{equation}
            where the bar denotes complex conjugation, e.g. $\alpha\in\Omega^{0,2}\oplus\overline{\Omega^{0,2}}$, then
            \begin{equation*}
                \alpha=\phi(z)\mathrm{d}z_1\wedge\mathrm{d}z_2+\overline{\phi(z)}\mathrm{d}\bar{z}_1\wedge\mathrm{d}\bar{z}_2
            \end{equation*}
            for coordinates $(z_1,z_2)\in\mathbb{C}^2$ (i.e. $\Omega^{0,2}\oplus\overline{\Omega^{0,2}}$ is real two-dimensional). $\Omega^{1,1}_{\parallel}$ denotes the real one-dimensional subspace of $\Omega^{1,1}$ along the K{\"a}hler form (which is hermitian). This follows from the fact that the Hodge star\footnote{The Hodge star on a complex manifold of complex dimension $n$ is $\bar{\star}:\Omega^{p,q}\rightarrow\Omega^{n-p,n-q}$ with $\bar{\star}\alpha=\star\bar{\alpha}$ and $\star$ the usual Hodge star extended to the complexification of $\Omega^{p+q}$.} $\bar{\star}$ acts on elements in $\Omega^{0,2},\Omega^{2,0}$ as complex conjugation (hence, $(1-\bar{\star})\Omega^{0,2}\oplus\overline{\Omega^{0,2}}=0$) and leaves the K{\"a}hler form invariant. We can ``unfold'' $\Omega^{1,1}_\parallel$ to level zero and obtain 
            \begin{equation}\label{eq--complexASD2}
                \begin{tikzcd}[column sep=scriptsize]
                    0\ar[r] & \Omega^{0,0}\oplus\overline{\Omega^{0,0}}\ar[r] & \Omega^{0,1}\oplus\overline{\Omega^{0,1}}\ar[r] & \Omega^{0,2}\oplus\overline{\Omega^{0,2}}\ar[r] & 0
                \end{tikzcd}.
            \end{equation}
            The complexification thereof, $(\Omega^\bullet,\mathrm{d}^+)_\mathbb{C}$, is simply given by the sum $(\Omega^{0,\bullet},\bar{\partial})\oplus(\Omega^{\bullet,0},\partial)$ of two Dolbeault complexes. On the other hand, we have
            \begin{equation}
                \Omega^{0,\bullet}\otimes(\mathcal{O}\oplus\Lambda^{2,0}T^\ast)=(\Omega^{0,\bullet}\otimes\mathcal{O})\oplus(\Omega^{0,\bullet}\otimes\Lambda^{2,0}T^\ast).
            \end{equation}
            Because the complexes are over $\mathbb{C}^2$ it is easy to write down explicitly an isomorphism for each summand separately, in such a way that they commute with the coboundary maps, giving rise to an isomorphism of complexes.
        \end{proof}

        \newtheorem*{prop--iso2}{Proposition 4.2}
        \begin{prop--iso2}
            {\it For the SD and ASD complexes, $(\Omega^\bullet,\mathrm{d}^+)$ and $(\Omega^\bullet,\mathrm{d}^-)$, on $\mathbb{C}^2$ there is an isomorphism}
            \begin{equation*}
                (\Omega^\bullet,\mathrm{d}^+)\simeq(\Omega^\bullet,\mathrm{d}^-),
            \end{equation*}
            {\it induced by the diffeomorphism $\mathbb{C}^2\ni(z_1,z_2)\mapsto(\bar{z}_1,z_2)$.}
        \end{prop--iso2}
        \begin{proof}(Sketch)
            We have seen in the (sketch) proof of \autoref{prop--iso1} that the SD complex is isomorphic to \eqref{eq--complexASD}.
            Complementary, the ASD complex is isomorphic \cite{Donaldson:1985} to
            \begin{equation}
                \begin{tikzcd}[column sep=scriptsize]
                    0\ar[r] & \Omega^{0}\ar[r] & \Omega^{0,1}\oplus\overline{\Omega^{0,1}}\ar[r] & \Omega^{1,1}_{\perp}\ar[r] & 0
                \end{tikzcd},
            \end{equation}
            where $\Omega^{1,1}_{\perp}$ denotes the real three-dimensional subspace of $\Omega^{1,1}$ orthogonal to the K{\"a}hler form. 
            Let $(z_1,z_2)$ denote the coordinates on $\mathbb{C}^2$ and consider the smooth map
            \begin{equation*}
                f:\mathbb{C}^2\rightarrow\mathbb{C}^2,\quad(z_1,z_2)\mapsto(\bar{z_1},z_2).
            \end{equation*}
            Clearly, $f$ is a diffeomorphism on $\mathbb{C}^2$ and induces a map $f^\ast$ acting on $\Omega^{1,1}_\perp$ by pullback. It can be verified explicitly by choosing bases for the respective spaces of forms that $f^\ast$ gives the sought-after isomorphism. For example, we have
            \begin{equation*}
                \Omega^{1,1}_\perp\ni\mathrm{i}(\mathrm{d}z_1\wedge\mathrm{d}\bar{z}_1-\mathrm{d}z_2\wedge\mathrm{d}\bar{z}_2)\longmapsto -\mathrm{i}(\mathrm{d}z_1\wedge\mathrm{d}\bar{z}_1+\mathrm{d}z_2\wedge\mathrm{d}\bar{z}_2)\in\Omega^{1,1}_\parallel
            \end{equation*}
            and similarly for the remaining two basis elements. As in the last proposition, it can then be verified by explicit computation that $f^\ast$ commutes with the coboundary maps.
        \end{proof}

    \section{$K$-Theory and the Symbol}\label{app--ktheory}
    
        In this appendix we introduce some basic notions of (topological) equivariant $K$-theory. In particular, we state why the symbol can be considered an element of the $K$-group and why this is relevant for the index computation.
        The exposition follows \cite{Landweber:2005,Segal:1968,Atiyah:1989} closely and we refer the interested reader to those references for a more detailed view on the subject.
        
        Let $X$ be a topological space which is compact and Hausdorff. The basic idea of $K$-theory is to probe topological properties of $X$ by considering complex vector bundles $E\overset{\pi}{\rightarrow}X$ of finite rank over $X$. The set of all such vector bundles is denoted $\vt(X)$. Since we are only interested in topological properties, we only concern ourselves with these vector bundles up to bundle isomorphisms and write $\vt_\simeq(X)$ for the quotient space. This can be made into a semi-group via the Whitney sum $\oplus$ (which descends to the quotient; the class of trivial bundles over $X$ is the identity). The $K$-group of $X$ is obtained by turning this semi-group into a group via the Grothendieck construction:

        \begin{defi}
            ($K$-group of $X$) The $K$-group of $X$ is defined as the quotient $K(X)=(\vt_\simeq(X)\times\vt_\simeq(X))/\sim$, where for all $E_1,E_2,F_1,F_2\in\vt_\simeq(X)$,
            \begin{equation*}
                (E_1,E_2)\sim(F_1,F_2):\Longleftrightarrow \exists G\in\vt_\sim(X): E_1\oplus F_2\oplus G=E_2\oplus F_1\oplus G.
            \end{equation*}
            The group action is given by $(E_1,E_2)\oplus(F_1,F_2)=(E_1\oplus F_1,E_2\oplus F_2)$.
        \end{defi}
        
        Intuitively, one might like to think of the equivalence class $[(E_1,E_2)]$ as the ``difference'' $E_1-E_2$ of the two vector bundles. Note that the $K$-group\footnote{To be more precise, we have defined the group $K^0(X)$, corresponding (under a natural transformation given by the Chern character) to even (rational) cohomology of $X$. There is also a group $K^1(X)$ corresponding to the odd part which, however, we will not be concerned with.} $K(X)$ is even a ring, by virtue of the tensor product $\otimes$ of vector bundles extending to the construction above.

        \begin{eg}\label{ex--Kgroup.point}
            The $K$-group over a point $\{pt\}$ is given by $K(\{pt\})\simeq\mathbb{Z}$. Every vector bundle over $\{pt\}$ is just a vector space. Up to isomorphism, those are classified by their dimension.
        \end{eg}

        We can also define maps between the $K$-groups of different spaces $X,Y$. Consider a continuous function $f:X\rightarrow Y$. Then for any vector bundle $E\in\vt(Y)$, $f$ induces the pullback bundle $f^\ast E\in\vt(X)$. It can be checked that this extends to a ring homomorphism\footnote{Hence, $K(\cdot)$ can be viewed as a contravariant functor from the category of compact topological spaces with continuous maps to the category of commutative unital rings with ring morphisms.} $f^\ast:K(Y)\rightarrow K(X)$. In particular, from \autoref{ex--Kgroup.point} we find $K(X)\rightarrow\mathbb{Z}$ for the inclusion of a point into $X$. 

        The attentive reader might have noticed that, in the main text, we always consider the $K$-group over tangent bundles $TX$, which are not compact, even if $X$ is. However, they are still locally compact and we can define their $K$-group in the following way:

        \begin{defi}
            ($K(X)$ for non-compact $X$) Let $X$ be locally compact. Then its $K$-group is defined by $K(X):=K(X^+)/K(\{pt\})$. Here $X^+$ is the one-point compactification of $X$.
        \end{defi}

        The relation of $K$-theory to the symbol of a complex of differential operators is established through the following

        \begin{thm}\label{thm--atiyah}
            Let $C^n(X)$ denote the set of compactly supported complexes of vector bundles over $X$ of length $n$, up to homotopy. Let $C^n_\emptyset(X)$ be the set of such complexes with empty support. Then, for $n\in\mathbb{N}$:
            \begin{equation*}
                K(X)\simeq C^n(X)/C^n_\emptyset(X).
            \end{equation*}
            $C^\infty(X)$ denotes the direct limit under inclusion $C^n(X)\subset C^{n+1}(X)$.
        \end{thm}
        \begin{proof}
            See \cite{Atiyah:1989} Theorem 2.6.1, p. 88.
        \end{proof}

        In order to appreciate the theorem above, we have to define the support of a complex:

        \begin{defi}
            (Support of a complex) The support of a complex $E^\bullet$ is the subset $\supp E^\bullet\subset X$ such that for $x\in\supp E^\bullet$, the sequence
            \begin{equation*}
                \begin{tikzcd}[column sep=scriptsize]
                    \dots\ar[r] & E^{n-1}|_{\pi^{-1}(x)}\ar[r] & E^n|_{\pi^{-1}(x)}\ar[r] & E^{n+1}|_{\pi^{-1}(x)}\ar[r] & \dots
                \end{tikzcd}
            \end{equation*}
            is not exact.
        \end{defi}
        
        In words, \autoref{thm--atiyah} allows to add to an existing complex an exact complex ``at no cost''. This is used, for example, when folding a complex. Although the resulting complex obviously differs from the original one, their support is identical and they belong to the same class in $K(X)$. 

        For a complex of differential operators
        \begin{equation}\label{eq--app.diffcomplex}
            \begin{tikzcd}[column sep=scriptsize]
                \dots\ar[r] & \Gamma(E^{n-1})\ar[r,"\mathrm{d}^{n-1}"] & \Gamma(E^n)\ar[r,"\mathrm{d}^n"] & \Gamma(E^{n+1})\ar[r] & \dots
            \end{tikzcd}
        \end{equation}
        with $\Gamma(E^n)$ denoting the space of sections on $E^n$, the corresponding symbol complex $\sigma(\mathrm{d})$ is given by
        \begin{equation}\label{eq--app.symcomplex}
            \begin{tikzcd}[column sep=scriptsize]
                \dots\ar[r] & \pi^\ast E^{n-1}\ar[r,"\sigma^{n-1}"] & \pi^\ast E^n\ar[r,"\sigma^n"] & \pi^\ast E^{n+1}\ar[r] & \dots
            \end{tikzcd}
        \end{equation}
        with $\pi:T^\ast X\rightarrow X$ the cotangent bundle and $\sigma^n$ bundle morphisms over $X$.
        By definition, if the complex \eqref{eq--app.diffcomplex} is elliptic, then it is exact outside of the zero-section $s_0:X\rightarrow T^\ast X,x\mapsto0_{T^\ast_xX}$. But $s_0\simeq X$ and $X$ is compact, thus, we see that $\sigma(\mathrm{d})\in C^\infty(X)$ (or for a definite length $n$ of $\sigma(\mathrm{d})$, $\sigma(\mathrm{d})\in C^n(X)$) and $[\sigma(\mathrm{d})]\in K(X)$. Hence, we can use the power of $K$-theory to analyze the symbol.

        Let us now move on to the equivariant case. Consider a compact Lie group $H$ acting on $X$ on the left via the map $H\times X\rightarrow X,(h,x)\mapsto h\cdot x$ with the usual conditions. This turns $X$ into an $H$-space.
        
        \begin{defi}
            ($H$-vector bundle) A vector bundle $\pi:E\rightarrow X$ over the $H$-space $X$ is called an $H$-vector bundle, if $E$ is an $H$-space such that
            \begin{enumerate}
                \item [(i)] $\pi$ respects the group action, i.e. $\pi\circ h=h\circ\pi$,
                \item [(ii)] the maps $E|_{\pi^{-1}(x)}\rightarrow E|_{\pi^{-1}(h\cdot x)}$ are linear maps for all $h\in H$.
            \end{enumerate}
        \end{defi}

        In complete analogy to the ordinary case, we can define the set $\vt_{\simeq,H}(X)$ of $H$-vector bundles of finite rank over $X$, up to isomorphisms, and apply the Grothendieck construction to get the equivariant $K$-group $K_H(X)$. Note that the Whitney sum and tensor product are defined in the ordinary way, turning $K_H(X)$ into a commutative unital ring.

        \begin{eg}\label{ex--KHgroup.point}
            The equivariant $K$-group over a point $\{pt\}$ is $K_H(\{pt\})\simeq R(H)$. Here, $R(H)$ is the representation ring of $H$, obtained by applying the Grothendieck construction to the semi-group of finite-dimensional complex representation spaces of $H$. This is a ring via the tensor product.
        \end{eg}

        \noindent Similarly to the ordinary case, continuous $H$-maps $f:X\rightarrow Y$ between $H$-spaces $X,Y$ induce homomorphisms $f^\ast:K_H(Y)\rightarrow K_H(X)$.

        Consider again the complex of differential operators \eqref{eq--app.diffcomplex} where now $E^n$ is an $H$-vector bundle. We can define an $H$-action on $s\in\Gamma(E^n)$ by $(h\cdot s)(x)=h\cdot(s(h^{-1}\cdot x))$. If the cochain maps $d^n$ commute with this $H$-action, we say that \eqref{eq--app.diffcomplex} is $H$-invariant. Then the cochain maps of the symbol \eqref{eq--app.symcomplex} are $H$-maps, i.e. the symbol also respects the $H$-action. Thus, we can define the set $C^n_H(X)$ of all compactly supported complexes of length $n$ of $H$-vector bundles over $X$ respecting the $H$-action, up to $H$-homotopy. 
        
        Finally, there is an analogue of \autoref{thm--atiyah}, saying that
        \begin{equation*}
            K_H(X)\simeq C^n_H(X)/C^n_{\emptyset,H}(X)
        \end{equation*} 
        For a proof, see \cite{Segal:1968} Proposition 3.1, p. 139. Hence, in particular, the symbol of an $H$-invariant elliptic complex is (a representative of) an element in $K_H(X)$.

\bibliographystyle{JHEP.bst}
\bibliography{main}

\end{document}